%% file: potts_arxiv.tex
\documentclass{article}

\input{macros}

\begin{document}

\title{%
	\huge
	Potts model with invisible colours \smallskip\\
	\large Random-cluster representation and {P}irogov-{S}inai analysis\medskip
	\footnotetext{Last update:~\today}
}

\author{%
	Aernout~C.~D. van~Enter$^{1,}$\footnote{%
		Email: \texttt{A.C.D.van.Enter@rug.nl}
	}
	\and
	Giulio Iacobelli$^{1,}$\footnote{%
		Email: \texttt{dajegiu@gmail.com}
	}
	\and
	Siamak Taati$^{1,2,}$\footnote{%
		Email: \texttt{siamak.taati@gmail.com}
	}
	\bigskip\\
	\small $^1$Johann Bernoulli Institute, University of Groningen, The Netherlands\\
	\small $^2$Mathematical Institute, Utrecht University, The Netherlands\\
}

\date{}

\maketitle

\begin{abstract}
	We study a variant of the ferromagnetic Potts model,
	recently introduced by Tamura, Tanaka and Kawashima,
	consisting of a ferromagnetic interaction
	among $q$ ``visible'' colours along with
	the presence of $r$ non-interacting ``invisible'' colours.
	We introduce a random-cluster representation for the model,
	for which we prove the existence of a first-order transition
    for any $q>0$, as long as $r$ is large enough.	
	When $q>1$, the low-temperature regime
	displays a $q$-fold symmetry breaking.  	
	The proof involves a Pirogov-Sinai analysis
	applied to this random-cluster representation of the model.
	
	\medskip
	
\end{abstract}

\tableofcontents

\section{Introduction}

Recently, in a series of papers~\cite{TamTanKaw10,TanTam10,TanTamKaw11},
Tamura, Tanaka and Kawashima introduced a variant of
the ferromagnetic Potts model to study the relation between
symmetry breaking and the order of the phase transition.
The model consists of a ferromagnetic Potts interaction
taking place between $q$ ``visible'' colours along with
the presence of $r$ ``invisible'' colours without any interaction.
They observed, through numerical simulations,
that in two dimensions with $q=2,3,4$
and $r$ large, the model undergoes a first-order phase transition
with $q$-fold symmetry breaking.
This is in contrast with the ordinary
two-dimensional $q$-colour Potts model with $q=2,3,4$,
in which the transition accompanied by the $q$-fold symmetry breaking
is known to be of second order~\cite{Bax82}.
This provides a simple example that a $q$-fold symmetry breaking in two dimensions
does not universally identify the order of the transition.

The transition in this model (as well as in the standard $q$-colour Potts model)
occurs from an \emph{ordered} state (which has a favoured direction among $q$ possibilities)
to a \emph{disordered} state (which has no favoured direction)
when the temperature is increased.
In the standard $q$-colour Potts model with $q$ small, this transition
is of \emph{second} order
(no latent heat at the transition point)
whereas in the Tamura-Tanaka-Kawashima version of the Potts model,
for the same values of $q$ but $r$ chosen sufficiently large,
the transition is of \emph{first} order
(the system absorbs heat during the transition, without changing its temperature).

For the standard $q$-colour Potts model, when $q$ is large enough,
there is a variety of different rigorous proofs
that the transition is of first order~\cite{KotShl82,BriKurLeb85,Mar86,LaaMesMirRuiShl91}
(see also~\cite{Kot94,Gri06}).
As announced	 in an earlier communication~\cite{EntIacTaa11},
in the present paper, we prove,
by minor adaptations of the proofs in~\cite{LaaMesMirRuiShl91,Kot94},
that when $q+r$ is large enough,
the Potts model with $q$ visible colours and $r$ invisible colours
undergoes a first-order phase transition.
The proof is based on an application of
the Pirogov-Sinai method to a random-cluster
representation of the model.


The phase transition could be better understood
if one thinks of a state of the system as
a possible resolution of the conflict between order and disorder.
The conflict should be resolved locally and in every region.
In an ordered region, the neighbouring sites tend to
take the same colour so as to minimize the energy, while
in a disordered region, the neighbouring sites take
their colours independently to maximize the entropy.
To establish the resolution of the order-disorder conflict,
one needs to take into account the disturbance
present at the interface between ordered and disordered regions (the \emph{contours}).


For the standard $q$-colour Potts model, the order-disorder
conflict is niftily depicted in the
Fortuin-Kasteleyn (a.k.a.~random-cluster)
representation of the model~\cite{ForKas72,Gri06,GeoHagMae00,Hag98}.
In this representation,
order is associated with the presence of bonds between
neighbouring sites and
disorder with the absence of bonds.
In the same spirit, we introduce a variant of
the Fortuin-Kasteleyn representation for the
Potts model with invisible colours.
The advantage of this new formulation is
that it admits a neat definition of
the interface between ordered and disordered regions.
Now, having two reference configurations
describing complete order and complete disorder ---
the one with every bond present, 
and the one with every bond absent 
---
as in~\cite{LaaMesMirRuiShl91},
we can apply the Pirogov-Sinai method~\cite{Kot94,Sin82,Zah84,BorImb89}.


In Section~\ref{sec:model}, we describe the model
and recall the formulation of first-order phase transition
in the Gibbsian setup.
Section~\ref{sec:rc} is dedicated to the introduction of a variant of
the random-cluster model and its connection
with the Potts model with invisible colours.
In Section~\ref{sec:contour:representation},
formal definitions for contours are provided, and it is shown
how to rewrite the partition functions of the model
in terms of contours.
These contour representations 
are then reduced, in Section~\ref{sec:contour:abstract},
to two abstract contour models, on which
standard techniques can be applied.
In Section~\ref{sec:damping},
starting from the two contour models,
we obtain two approximations
for the free energy of the Potts model with invisible colours.
If $q+r$ is large,
each of these approximations turns out to be accurate
in an interval of temperatures,
one whenever order prevails 
and the other when disorder is dominant.
The two intervals exhaust all the temperatures
and have a unique common point, which is
the transition point of the system.
Finally, the above two approximations
are used in Section~\ref{sec:main}
to prove a first-order transition
at the transition point.
The occurrence of the symmetry breaking at the same
transition point then follows, using
standard properties of the random-cluster representation,
which are reviewed in Appendix~\ref{apx:rc:properties}.

\paragraph{Acknowledgments.}
G.I. and S.T. thank NWO for support. We thank Roberto Fern\'andez for helpful conversations.



\section{Potts Model with Invisible Colours}
\label{sec:model}
\subsection{The Model}

Let $\xL$ denote the two-dimensional square \emph{lattice}, which
we think of as a graph $(\xS,\xB)$,
where $\xS$ denotes the set of sites (identified by $\xZ^2$)
and $\xB$ the set of nearest neighbour bonds.
In the \emph{$(q,r)$-Potts model},
each site $i\in \xS$ is in one of $(q+r)$ colours
$1,2,\ldots,q,q+1,\ldots,q+r$.
Therefore, a \emph{configuration} of the model is an assignment of
values from the set $\{1,2,\ldots,q,q+1,\ldots,q+r\}$ to
the sites in $\xS$.
The $(q,r)$-Potts model~\cite{TamTanKaw10,TanTam10,TanTamKaw11}
is expressed by the formal Hamiltonian
\begin{align}
\label{eq:hamiltonian:def}
	H(\sigma) &=
		-\sum_{\{i,j\}\in \xB}
			\delta(\sigma_i=\sigma_j\leq q) \;,
\end{align}
where $\delta(\sigma_i=\sigma_j\leq q)$ is $1$ if $\sigma_i=\sigma_j\leq q$ and $0$ otherwise.
Each pair of neighbouring sites that have the same colour $\alpha\leq q$
contributes with energy $-1$, while sites with colours $\alpha>q$
do not contribute to the energy.
The first $q$ colours are hence called the \emph{visible} colours,
and the rest the \emph{invisible} colours.
If there are no invisible colours (i.e, if $r=0$),
the model reduces to the ordinary Potts model with $q$ colours.
As in the ordinary $q$-colour Potts model, the $(q,r)$-Potts model
has precisely $q$ periodic ground-state configurations,
in which every site has the same visible colour.

Following the usual approach,
we describe the system in thermal equilibrium
via probability distributions on the space of all
possible configurations of the model.
The \emph{Boltzmann distribution} on a finite volume
$\Lambda\subseteq\xL$ with boundary condition $\omega$
at inverse temperature $\beta$ is defined by
\begin{align}
	\mu_{\beta,\Lambda}^\omega(\sigma_\Lambda) &=
		\frac{1}{Z_\beta^\omega(\Lambda)}
			\xe^{-\beta H_\Lambda(\sigma_\Lambda\omega_{\Lambda^\complement})} \;,
\end{align}
where $H_\Lambda(\sigma_\Lambda\omega_{\Lambda^\complement})$
consists of a finite number of terms in the formal Hamiltonian~(\ref{eq:hamiltonian:def})
corresponding to the energy of $\sigma_\Lambda$ and its interaction
with the boundary condition $\omega$.
Namely, 
\begin{align}
	H_\Lambda(\sigma) &=
		-\sum_{\substack{\{i,j\}\in \xB \\ \makebox[3em]{\scriptsize $i\in S(\Lambda)$ or $j\in S(\Lambda)$}}}
			\delta(\sigma_i=\sigma_j\leq q) \;,
\end{align}
which can be decomposed as a sum
\begin{align}
	H_\Lambda(\sigma_\Lambda\omega_{\Lambda^\complement}) &=
		H_\Lambda^{\mathrm{int}}(\sigma_\Lambda) +
		H_\Lambda^{\mathrm{bound}}(\sigma_\Lambda\omega_{\Lambda^\complement}) \;,
\end{align}
where $H_\Lambda^{\mathrm{int}}(\sigma_\Lambda)$ involves
the interaction terms within $\Lambda$, and
$H_\Lambda^{\mathrm{bound}}(\sigma_\Lambda\omega_{\Lambda^\complement})$
represents the terms corresponding to the interaction
of $\Lambda$ with its boundary.
%
%
%
%
%
The factor $Z_\beta^\omega(\Lambda)$ is a normalizing constant
--- the \emph{partition function} ---
making $\mu_{\beta,\Lambda}^\omega$ a probability distribution.
More specifically,
the partition function of volume $\Lambda$ with boundary condition $\omega$
is given by
\begin{align}
\label{eq:potts:partition:boundary}
	Z_\beta^\omega(\Lambda) &=
		\sum_{\sigma_\Lambda}
		\xe^{
		-\beta H_\Lambda^{\mathrm{int}}(\sigma_\Lambda) -
		\beta H_\Lambda^{\mathrm{bound}}(\sigma_\Lambda\omega_{\Lambda^\complement})
		} \;.
\end{align}
If we ignore the boundary term, we obtain 
the \emph{free-boundary} partition function of volume $\Lambda$:
\begin{align}
	Z_\beta^\free(\Lambda) &=
		\sum_{\sigma_\Lambda}
		\xe^{-\beta H_\Lambda^{\mathrm{int}}(\sigma_\Lambda)} \;.
\end{align}
This is the normalizing factor for the
free-boundary Boltzmann distribution on $\Lambda$.
A \emph{Gibbs measure} on the space of all configurations of
the infinite lattice system, at inverse temperature~$\beta$,
is a probability measure $\mu$ whose conditional probabilities
for every finite volume $\Lambda$,
given the configuration~$\omega$ outside $\Lambda$,
are given by the Boltzmann distribution $\mu_{\beta,\Lambda}^\omega$.
More specifically,
\begin{align}
	\mu\left(A \text{ and } B\right) &=
		\int_B \mu_{\beta,\Lambda}^\omega(A)\mu(\xd\omega) \;,
\end{align}
for every event $A$ not depending on the colours of
the sites outside $\Lambda$ and
every event $B$ not depending on the colours of
the sites in $\Lambda$ (see e.g.~\cite{Geo88}).
It follows from a compactness argument that
such measures exist at every temperature.
However, when the temperature is sufficiently low,
it is possible to have several distinct Gibbs measures.
The multiplicity of Gibbs measures is then interpreted
as the possibility of co-existence of distinguishable phases
of the physical system (in this case,
the possibility of spontaneous magnetization
in $q$ different directions).  We refer to~\cite{Geo88}
for details.

One way to obtain Gibbs measures consists of
taking the thermodynamic limit of the Boltzmann distribution
with or without a fixed boundary condition.  
For a visible colour $k$, let $\omega^k$ denote
the configuration of the lattice in which every site
has colour $k$.  Let $\mu^k_\beta$ denote a Gibbs measure
obtained by taking a weak limit of finite-volume
Boltzmann distributions with boundary condition~$\omega^k$
at inverse temperature $\beta$,
when the finite-volume grows to the whole lattice.
Similarly, we obtain a Gibbs measure $\mu^\free_\beta$
by taking a weak limit of free-boundary Boltzmann distributions.


For every $n>0$, let $\Lambda_n$ denote the
$(2n+1)\times(2n+1)$ central square in the lattice, which we see
as the subgraph of the lattice induced by the sites in $[-n,n]^2$.
The \emph{pressure} of the model is defined by
\begin{align}
\label{eq:potts:pressure:def}
	f(\beta) &=
		\lim_{n\to\infty} 
		\frac{1}{\abs{S(\Lambda_n)}}\log Z_\beta^\omega(\Lambda_n) \;,
\end{align}
in which $S(\Lambda_n)$ denotes the set of sites in $\Lambda_n$.
The function $-\frac{1}{\beta}f(\beta)$ is the \emph{free energy} per site.
The limit exists and is independent of the boundary condition $\omega$
(see e.g.~\cite{Isr79,Rue04}).  We would also get the same limit
as in~(\ref{eq:potts:pressure:def})
if we used the free boundary partition function.
The particular choice of volumes used above is not crucial,
and can be replaced by any sequence
satisfying the van Hove property~(see~\cite{Isr79}).


\subsection{First-Order Phase Transition}

A \emph{first-order phase transition} in temperature is characterized
by the presence of \emph{latent heat} at the transition point~\cite{Dor99}.
This means that at the transition point, the system
absorbs or gives out heat without a change in temperature.
The presence of latent heat, therefore, corresponds to a jump
in the internal energy.

In the Gibbsian setup, the state of a system in thermal equilibrium
is represented by a Gibbs measure (see~\cite{Geo88}).
If the Gibbs measure is translation-invariant,
the \emph{internal energy} density of the system is described by
the expected value of energy per site.  The presence of
latent heat at a temperature means that the limits of the internal energy
from above and below the transition temperature are different.
This implies, by continuity, the existence of
two translation-invariant Gibbs measures at that temperature
having different internal energy.

If the pressure function $f(\beta)$ is differentiable at
a point $\beta$, its derivative at $\beta$ coincides
with the internal energy
with respect to every translation-invariant Gibbs measure
at $\beta$ (see~\cite{Isr79,Rue04}).
(This, however, does not rule out the possibility
of the existence of several translation-invariant Gibbs measures
at~$\beta$.)
If, on the other hand, the pressure function $f(\beta)$ is
non-differentiable at a point $\beta$,
its left and right derivatives at $\beta$
(which exist due to convexity)
are different and coincide with the internal energy
with respect to two different translation-invariant
Gibbs measures at $\beta$.
The difference between these two derivatives
corresponds to a latent heat at $\beta$,
implying that the system undergoes a first-order
phase transition at $\beta$.

In this paper, we show the existence of
a first-order phase transition in the $(q,r)$-Potts model
by proving that the pressure function $f(\beta)$
has a unique non-differentiable point $\beta_\critical$.
Above~$\beta_\critical$, the system admits
$q$ ``ordered'' translation-invariant Gibbs measures.
Each ordered measure can be thought of as a perturbation
of one of the $q$ ground-state configurations, in the sense that
with probability~$1$, the configuration of the model consists of
a unique infinite sea of one of the visible colours with
finite islands of disturbance.
Below~$\beta_\critical$, the system has a
``disordered'' translation-invariant Gibbs measure,
which can be seen as a perturbation of the
uniform Bernoulli measure: there is a unique infinite
sea of independent colours with finite islands of disturbance.
At~$\beta_\critical$, the $q$~ordered measures co-exist
with the disordered one.

\section{Biased Random-Cluster Representation}
\label{sec:rc}


In analogy with the standard Potts model~\cite{ForKas72,Gri06,GeoHagMae00},
it is possible to rewrite the partition function for
the $(q,r)$-Potts model in terms of the partition function
for a variant of the random-cluster model.
While the former is a model defined on sites,
the latter will be a model defined on bonds.
The random-cluster representation of the Potts model
allows for an elegant formulation of the intuitive
concepts of ``order'' and ``disorder'':
the presence of a bond in the random-cluster representation
is interpreted as ``order'', while
the absence of a bond as ``disorder''~\cite{LaaMesMirRuiShl91}.

Although for the purpose of our problem, it suffices
to present the connection for squares $\Lambda_n$ in the lattice,
we elucidate the connection for an arbitrary finite graph,
where there is no boundary condition.
Later, we explain how the boundary conditions affect this connection.

Let $\xG=(S,B)$ be a finite graph.
The \emph{$r$-biased random-cluster} model on $\xG$ is given by a
probability distribution on the sets $X\subseteq B$.
The distribution 
has three parameters $0\leq p\leq 1$, $q>0$ and $r\geq 0$,
and is defined by
\begin{align}
	\phi_{p,q,r}(X) &=
		\frac{1}{Z^{\RC}_{p,q,r}(\xG)}
		\left[\prod_{b\in B}
			p^{\delta(b\in X)}(1-p)^{\delta(b\notin X)}\right]
			(q+r)^{\kappa_0(S,X)}q^{\kappa_1(S,X)} \;,
\end{align}
in which $\kappa_0(S,X)$ denotes the number of
isolated sites of the graph $(S,X)$
and $\kappa_1(S,X)$ the number of
non-singleton connected components of $(S,X)$ and
$Z^{\RC}_{p,q,r}(\xG)$ the partition function.
Notice that for $r=0$, the model reduces to the standard
random-cluster model, in which both singleton and non-singleton
connected components have weight $q$.
For $r>0$, the above model induces a bias
towards singleton connected components.
Namely, the singleton connected components have weight $(q+r)$
whereas the non-singleton connected components have weight $q$.

Let us now see how the $(q,r)$-Potts model is related
to the $r$-biased random-cluster model.
This is a mere generalization of the standard relation
between the Potts and random-cluster models~\cite{Gri06,GeoHagMae00}.
Let $\Omega$ be the set of $(q,r)$-Potts configurations on $\xG$.
The partition function of this model can be rewritten as
\begin{align}
	Z_\beta(\xG) &=
		\sum_{\sigma\in\Omega}\xe^{
			\beta\sum_{\{i,j\}\in B} \delta(\sigma_i=\sigma_j\leq q)
		} \\
	&=
		\sum_{\sigma\in\Omega} \prod_{\{i,j\}\in B}
			\xe^{\beta\delta(\sigma_i=\sigma_j\leq q)} \\
	&=
		\sum_{\sigma\in\Omega} \prod_{\{i,j\}\in B}
			\left[1+\delta(\sigma_i=\sigma_j\leq q)(\xe^\beta-1)\right] \\
	&=
		\sum_{\sigma\in\Omega}\sum_{X\subseteq B}
			(\xe^\beta-1)^{\abs{X}}
			\prod_{\{i,j\}\in X}\delta(\sigma_i=\sigma_j\leq q)
			 \\
	\label{eq:coupling}
	&=
		\sum_{\sigma\in\Omega}\sum_{X\subseteq B}
			\pi(\sigma,X) \;,
\end{align}
where
\begin{align}
	\pi(\sigma,X) &=
		\xe^{\beta\abs{B}}
		\prod_{\{i,j\}\in B}
		\left[
			\delta(\{i,j\}\in X)\delta(\sigma_i=\sigma_j\leq q)(1-\xe^{-\beta}) +
			\delta(\{i,j\}\notin X)\xe^{-\beta}		
		\right] \;.
\end{align}
The latter expression can be seen as a coupling of the
$(q,r)$-Potts distribution on $\Omega$
and a probability distribution on the space $\{0,1\}^{B}$.
The marginal of this coupling on $\{0,1\}^{B}$
is simply the $r$-biased random-cluster distribution
$\phi_{p_\beta,q,r}$ with $p_\beta=1-\xe^{-\beta}$.
In particular, the weight $\pi(\sigma,X)$ can also be expressed as
\begin{align}
	\pi(\sigma,X) &=
		\xe^{\beta\abs{B}}\cdot 1_{F_r}(\sigma,X)\cdot
		\prod_{\{i,j\}\in B}
		\left[
			p_\beta\,\delta(\{i,j\}\in X) +
			(1-p_\beta)\,\delta(\{i,j\}\notin X)
		\right] \;,
\end{align}
where
\begin{align}
	F_r &\IsDef\left\{
		(\sigma,X): \sigma_i=\sigma_j\leq q
		\text{ for all } \{i,j\}\in X
	\right\} \;.
\end{align}
The effect of the bias in the $r$-biased random-cluster model
reduces to an increase in the number of compatible
configurations with a given $X$, which is driven
by a larger number of choices for the colour of those sites constituting the
singleton connected components.
In short, for each $X\subseteq B$, we have
\begin{align}
	\sum_{\sigma\in\Omega} 1_{F_r}(\sigma,X) &=
		q^{\kappa_1(S,X)}(q+r)^{\kappa_0(S,X)} \;.
\end{align}

The above coupling could be interpreted in either of the following ways~\cite{Gri06,GeoHagMae00}:
\begin{enumerate}[ I.]
	\item We first sample $\sigma$ according to the $(q,r)$-Potts distribution.
		Then, we choose the elements of $X$ from $B$, randomly and independently, as follows:
		for each bond $\{i,j\}\in B$ with $\sigma_i=\sigma_j$, we put $\{i,j\}$ in $X$
		with probability $p_\beta$;  for each bond $\{i,j\}\in B$ with $\sigma_i\neq\sigma_j$,
		we do not put $\{i,j\}$ in $X$.
	\item We first sample $X$ according to the $r$-biased random-cluster distribution
		$\phi_{p_\beta,q,r}$.  Then, for each non-singleton connected component of $(S,X)$,
		we pick a random colour uniformly among the visible colours, and colour
		every site in the component with that colour.
		Last, for every isolated site in $(S,X)$, we choose a random colour uniformly
		among all the possible colours.
		(The choices of colours ought to be independent of each other.)
\end{enumerate}

We can now use~(\ref{eq:coupling}) to obtain
\begin{align}
\label{eq:potts-rc:finite}
	Z_\beta(\xG) &= \xe^{\beta\abs{B}} Z^{\RC}_{p_\beta,q,r}(\xG)
\end{align}
with $p_\beta=1-\xe^{-\beta}$.


For finite subgraphs of the infinite lattice,
we will be using only two types of partition functions for
the $(q,r)$-Potts model, namely
the one with free boundary and the ones with homogenous boundary conditions.
In the following, we see how the above two types of boundary conditions
translate into the so-called \emph{disordered} and \emph{ordered} boundary conditions
for the $r$-biased random-cluster model.
Although, setting $\xG=\Lambda_n$,
equation~(\ref{eq:potts-rc:finite}) already provides a relation between
the free-boundary partition functions of
the two models,
we will work with a slightly different relation,
connecting the free-boundary partition function of the $(q,r)$-Potts model
to a partition function for the $r$-biased random-cluster model
that involves a boundary condition.
This new relation will turn out to be more convenient in the sequel.

The free-boundary partition function for the $(q,r)$-Potts model can be written as
\begin{align}
\label{eq:potts-rc:disordered}
	Z_\beta(\Lambda_n) &=
		(q+r)^{-\abs{S(\Lambda_{n+1})\setminus S(\Lambda_n))}}
		\cdot\xe^{\beta\abs{B(\Lambda_{n+1})}}\cdot
		Z^{\RC.\disorder}_{p_\beta,q,r}(\Lambda_{n+1}) \;,
\end{align}
where $Z^{\RC.\disorder}_{p_\beta,q,r}(\Lambda_{n+1})$
is the partition function with \emph{disordered boundary condition}
for the $r$-biased random-cluster model.
The latter is defined by
\begin{align}
\label{eq:rc-partition-function:disordered}
	Z^{\RC.\disorder}_{p_\beta,q,r}(\Lambda_{n+1}) &=
		\sum_{X\in \mathcal{X}^\disorder_{\Lambda_{n+1}}}
			p_\beta^{\abs{X}}
			(1-p_\beta)^{\abs{B(\Lambda_{n+1})\setminus X}}
			(q+r)^{\kappa_0(S(\Lambda_{n+1}),X)}
			q^{\kappa_1(S(\Lambda_{n+1}),X)} \;,
\end{align}
where
\begin{align}
	\mathcal{X}^\disorder_{\Lambda_{n+1}} &= \left\{
		X\subseteq B(\Lambda_{n+1}):
			X\cap\left(B(\Lambda_{n+1})\setminus B(\Lambda_n)\right) =\varnothing
	\right\} \;.
\end{align}
Similarly, for the boundary condition $\omega^k$ we get
\begin{align}
\label{eq:potts-rc:ordered}
	Z_\beta^{\omega^k}(\Lambda_n) &=
		q^{-1}\cdot(\xe^\beta-1)^{-\abs{B(\Lambda_{n+1}\setminus\Lambda_n)}}		
		\cdot\xe^{\beta\abs{B(\Lambda_{n+1})}}\cdot
		Z^{\RC.\order}_{p_\beta,q,r}(\Lambda_{n+1}) \;,
\end{align}
where $Z^{\RC.\order}_{p_\beta,q,r}(\Lambda_{n+1})$
is the partition function with \emph{ordered boundary condition}
for the $r$-biased random-cluster model,
which is defined by
\begin{align}
\label{eq:rc-partition-function:ordered}
	Z^{\RC.\order}_{p_\beta,q,r}(\Lambda_{n+1}) &=
		\sum_{X\in \mathcal{X}^\order_{\Lambda_{n+1}}}
			p_\beta^{\abs{X}}
			(1-p_\beta)^{\abs{B(\Lambda_{n+1})\setminus X}}
			(q+r)^{\kappa_0(S(\Lambda_{n+1}),X)}
			q^{\kappa_1(S(\Lambda_{n+1}),X)} \;,
\end{align}
where
\begin{align}
	\mathcal{X}^\order_{\Lambda_{n+1}} &= \left\{
		X\subseteq B(\Lambda_{n+1}):
			X\supseteq B(\Lambda_{n+1}\setminus\Lambda_n)		
	\right\} \;.
\end{align}
By $\Lambda_{n+1}\setminus\Lambda_n$ we mean the graph obtained from $\Lambda_{n+1}$
by removing all the sites in $\Lambda_n$ and the bonds attached to them.
Let us remark that although mathematically
$\mathcal{X}^\disorder_{\Lambda_{n+1}}$
is simply the collection of all subsets $X\subseteq B(\Lambda_n)$,
we wrote it as above to emphasize that the elements of
$\mathcal{X}^\disorder_{\Lambda_{n+1}}$
are configurations of $B(\Lambda_{n+1})$.
See Figure~\ref{fig:typical-rc-configs} for typical examples of elements in
$\mathcal{X}^\disorder_{\Lambda_{n+1}}$ and $\mathcal{X}^\order_{\Lambda_{n+1}}$.

\begin{figure}
	\begin{center}
		\begin{tabular}{c@{\qquad\qquad\qquad}c}
			\begin{tikzpicture}[scale=.4,line cap=round]
				\newdimen\noderadius
				\noderadius=1.1mm
				
				\def\volsize{10}
				\def\db{0.25}
				
				\draw (5.3,11) node {\small $\Lambda_{n+1}$};
				
				\draw[dotted] (-\db,-\db) rectangle ({\volsize+\db},{\volsize+\db});
				\draw[dotted] ({1-\db},{1-\db}) rectangle ({\volsize+\db-1},{\volsize+\db-1});

				\draw[gray,very thin] (-\db,-\db) grid ({\volsize+\db},{\volsize+\db});

				\foreach \i in {1,...,9} {
					\foreach \j in {1,...,8} {
						\pgfmathrandominteger{\a}{1}{2}
						\ifnum \a<2
							\draw[very thick] (\i,\j) -- (\i,\j+1);
						\fi
					}
				}
				\foreach \j in {1,...,9} {
					\foreach \i in {1,...,8} {
						\pgfmathrandominteger{\a}{1}{2}
						\ifnum \a<2
							\draw[very thick] (\i,\j) -- (\i+1,\j);
						\fi
					}
				}

			\end{tikzpicture}
		
		&

			\begin{tikzpicture}[scale=.4,line cap=round]
				\newdimen\noderadius
				\noderadius=1.1mm
				
				\def\volsize{10}
				\def\db{0.25}
				
				\draw (5.3,11) node {\small $\Lambda_{n+1}$};
				
				\draw[dotted] (-\db,-\db) rectangle ({\volsize+\db},{\volsize+\db});
				\draw[dotted] ({1-\db},{1-\db}) rectangle ({\volsize+\db-1},{\volsize+\db-1});

				\draw[gray,very thin] (-\db,-\db) grid ({\volsize+\db},{\volsize+\db});
				
				\draw[very thick] (0,0) rectangle (\volsize,\volsize);

				\foreach \i in {1,...,9} {
					\foreach \j in {0,...,9} {
						\pgfmathrandominteger{\a}{1}{2}
						\ifnum \a<2
							\draw[very thick] (\i,\j) -- (\i,\j+1);
						\fi
					}
				}
				\foreach \j in {1,...,9} {
					\foreach \i in {0,...,9} {
						\pgfmathrandominteger{\a}{1}{2}
						\ifnum \a<2
							\draw[very thick] (\i,\j) -- (\i+1,\j);
						\fi
					}
				}

			\end{tikzpicture}

		\medskip \\
		
		(a) & (b)
		\end{tabular}
		\end{center}
	\caption{(a) A configuration in $\mathcal{X}^\disorder_{\Lambda_{n+1}}$.
		(b) A configuration in $\mathcal{X}^\order_{\Lambda_{n+1}}$.}
	\label{fig:typical-rc-configs}
\end{figure}
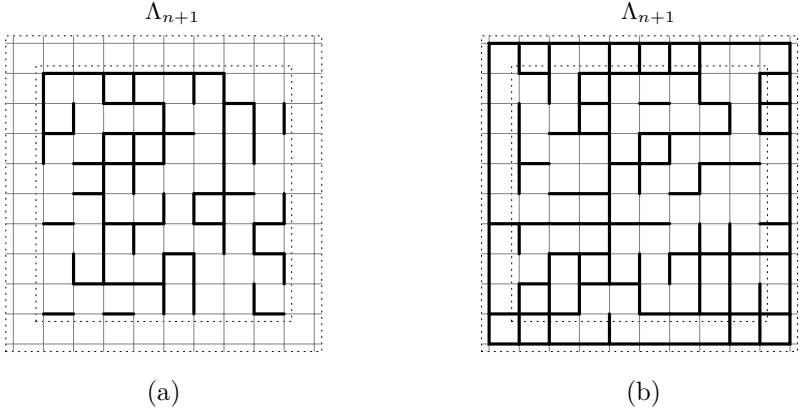

In the following section, we will extend the definition of
$Z^{\RC.\disorder}_{p_\beta,q,r}$ and $Z^{\RC.\order}_{p_\beta,q,r}$ to
arbitrary subgraphs of the lattice.


Using the above relationships, we obtain that the pressure
of the $(q,r)$-Potts model can be written as
\begin{align}
	f(\beta) &=
		2\left(\beta + \lim_{n\to\infty}\frac{\log Z^{\RC.\disorder}_{p_\beta,q,r}(\Lambda_n)}{\abs{B(\Lambda_n)}}\right) =
		2\left(\beta + \lim_{n\to\infty}\frac{\log Z^{\RC.\order}_{p_\beta,q,r}(\Lambda_n)}{\abs{B(\Lambda_n)}}\right) \;.
\end{align}
We call the limit
\begin{align}
	f^\RC(\beta) &=
		\lim_{n\to\infty}\frac{\log Z^{\RC.\disorder}_{p_\beta,q,r}(\Lambda_n)}{\abs{B(\Lambda_n)}} =
		\lim_{n\to\infty}\frac{\log Z^{\RC.\order}_{p_\beta,q,r}(\Lambda_n)}{\abs{B(\Lambda_n)}}
\end{align}
the \emph{pressure} (per bond) of the $r$-biased random-cluster representation.
Note that the singularities of the $(q,r)$-Potts pressure function $f(\beta)$
can be detected by studying the pressure function $f^\RC(\beta)$.
One advantage of this random-cluster representation is that it has a
more transparent expression in terms of ``contours'', which helps us
study the function $f^\RC(\beta)$.

\section{Reduction to Contour Model}
\label{sec:contour}
\subsection{Contour Representation}
\label{sec:contour:representation}

Any configuration of the $r$-biased random-cluster model
in a volume is a subset $X$ of bonds in the volume.
We interpret each bond in $X$ as ``ordered''
and each bond outside $X$ as ``disordered''.
Any configuration $X$ can then be seen as clusters of ordered
and disordered bonds.
Whether an equilibrium state is ordered or disordered
can be seen as the result of a competition between
ordered and disordered regions.
The selection criterion for this competition is ``energy''.
The term ``energy'' refers to an abstract notion of energy for
the $r$-biased random-cluster model, which in analogy with
the Boltzmann distribution, corresponds to minus logarithm of probability.

Let us define the ``energy'' of an ordered bond
as the ``energy'' per bond of the fully ordered configuration; that is,
\begin{align}
	e(\xB) &\IsDef
		\lim_{n\to\infty}\frac{%
			-\log\left[q(1-\xe^{-\beta})^{\abs{B(\Lambda_n)}}\right]
		}{\abs{B(\Lambda_n)}} =
		-\log(1-\xe^{-\beta}) \;.
\end{align}
Similarly, define the ``energy'' of a disordered bond as
the ``energy'' per bond of the fully disordered configuration:
\begin{align}
	e(\varnothing) &\IsDef
		\lim_{n\to\infty}\frac{%
			-\log\left[\xe^{-\beta\abs{B(\Lambda_n)}}(q+r)^{\abs{S(\Lambda_n)}}\right]
		}{\abs{B(\Lambda_n)}} =
		-\log\left[\xe^{-\beta}\sqrt{q+r}\,\right] \;.
\end{align}
The ``energy'' of the ordered and the disordered regions
can now be expressed as
$\abs{R^\xorder}\cdot e(\xB)$ and $\abs{R^\xdisorder}\cdot e(\varnothing)$,
respectively,
where $\abs{R^\xorder}$ and $\abs{R^\xdisorder}$ denote the
size of the ordered and disordered regions.
The ``energy'' of $X$, in turn, can be written in terms of
the ``energy'' of ordered and disordered regions
plus a correction term due to
the effect of the boundaries separating them.
If the effect of these boundaries is negligible
(which will turn out to be the case whenever $q+r$ is large),
the selection criterion for the competition between
order and disorder boils down to determining which
of~$e(\xB)$ and~$e(\varnothing)$ is minimal.
This is the starting point of the Pirogov-Sinai approach
to study phase transitions (see e.g.~\cite{Kot94}).

The presence of the correction term at the boundaries
can be explained as follows:
In the probability weight of a configuration~$X$,
each isolated site contributes with a factor~$(q+r)$.
To express the ``energy'' of the disordered regions
purely in terms of bonds, we evenly distribute the contribution
of the isolated sites among the~$4$
incident bonds.  Doing so, every disorder bond acquires
zero, one or two ``energy''-shares,
depending on the number of isolated sites
it is incident to.
Since in the fully disordered configuration $\varnothing$
there is no ordered region, every bond is incident to
precisely two isolated sites
and receives
two ``energy''-shares, leading to the factor $(q+r)^{\frac{2}{4}}$
in the expression of~$e(\varnothing)$.
In an arbitrary configuration, however,
the disordered bonds on the borderline between the ordered
and disordered regions, receive one or no ``energy''-share,
hence the need for a correction term.

It is possible to define a suitable notion of boundary
between ordered and disordered regions,
so that each configuration $X$ is uniquely identified
by its boundary (see below).
We could then rewrite the partition functions as sums
running over ``admissible'' boundaries, that is,
those corresponding to configurations of bonds.
Each admissible boundary is split into ``primary'' objects
termed contours whose ``energy'' add up to
the corresponding boundary effect.

In the following, we specify rigorously
the above heuristic notions of ``boundary'' and ``contours''.
We define the \emph{boundary} of a configuration $X\subseteq \xB$
as the set
\begin{align}
	\xpartial X &\IsDef
		\{(i,b)\in\xS\times\xB: i\sim b \text{ and } i\in S(X) \text{ and } b\notin X\} \;,
\end{align}
where $i\sim b$ means site $i$ and bond $b$ are incident,
and $S(X)$ is the set of sites incident to bonds in~$X$.
The set $\xpartial X$ uniquely determines $X$.
We say that two bonds $b$ and $b'$ in the lattice are \emph{co-adjacent}
if they belong to the same unit square.
More intuitively, co-adjacency is equivalent to adjacency in the dual lattice.
%
A set of bonds $X$ is \emph{co-connected} if for every
two bonds $b,b'\in X$, there is a sequence
$b=b_1,b_2,\ldots,b_n=b'$ of bonds in $X$ such that $b_i$ and $b_{i+1}$
are co-adjacent.
A \emph{contour} is a set $\gamma\subseteq \xS\times \xB$ such that
\begin{enumerate}[ i)]
	\item the set of bonds appearing in $\gamma$,
		denoted by $B(\gamma)$, is co-connected, and
	\item there exists a configuration $X$ such that $(S(X),X)$ is connected
		and $\gamma=\xpartial X$.
\end{enumerate}
We shall denote by $\Gamma$ the set of all finite contours in $\xL$.
If $\gamma$ is a contour, then removing the bonds $B(\gamma)$
breaks the lattice $\xL$ into connected subgraphs.
If $\gamma$ is finite, the graph 
$\xL\setminus B(\gamma)$
has a unique infinite connected component, which we call
the \emph{exterior} of $\gamma$ and denote by $\exterior\gamma$.
The subgraph $\xL\setminus B(\gamma)\setminus\exterior\gamma$
(which could be empty or disconnected) is called
the \emph{interior} of $\gamma$ and is denoted by $\interior\gamma$.
By $V(\gamma)$ we will mean the union of $\interior\gamma$
and the subgraph induced by $B(\gamma)$.\footnote{%
	By the subgraph induced by a set of bonds we mean
	the graph obtained by those bonds and their endpoints.
}
Let $\gamma$ be a finite contour. The configuration
$X$ such that $(S(X),X)$ is connected and $\gamma=\xpartial X$
(which exists by definition) is either finite or co-finite.
If $X$ is finite, we call $\gamma$ a \emph{disorder} contour,
and if $X$ is co-finite, we call $\gamma$ an \emph{order} contour.
Note that, if $\gamma$ is a disorder contour, all the sites
appearing in $\gamma$ are in the interior of $\gamma$,
whereas if $\gamma$ is an order contour, all the sites
appearing in $\gamma$ are in the exterior of $\gamma$.
As a result, we can safely represent a finite contour $\gamma$
by the pair $(B(\gamma),x)$ where
$x$ is a label specifying the type of the contour (disorder or order).
This also means that the set of all finite contours $\Gamma$ can be partitioned
into two subsets: the set of disorder contours, which we denote
by~$\Gamma^\xdisorder$, and the set of order contours, which we
denote by~$\Gamma^\xorder$.

\begin{figure}
	\begin{center}
		\begin{tabular}{c@{\qquad\qquad\qquad}c}
			\begin{tikzpicture}[scale=.4,line cap=round]
				\newdimen\noderadius
				\noderadius=1.5mm
				\def\nodecol{white}
				\def\bondcol{black}
				
				\def\volsize{8}
				\def\db{0.25}

				\def\xboundtail{0.9}
				\def\xbound[#1]#2{%
					\ifthenelse{\equal{#1}{l}}{
						\coordinate (A) at #2;
						\foreach \r in {0.2,0.5,...,\xboundtail}
							\draw[\bondcol] (A)+(-\r,0) node {\tiny $\star$};
					}{}
					\ifthenelse{\equal{#1}{u}}{
						\coordinate (A) at #2;
						\foreach \r in {0.2,0.5,...,\xboundtail}
							\draw[\bondcol] (A)+(0,\r) node {\tiny $\star$};
					}{}
					\ifthenelse{\equal{#1}{r}}{
						\coordinate (A) at #2;
						\foreach \r in {0.2,0.5,...,\xboundtail}
							\draw[\bondcol] (A)+(\r,0) node {\tiny $\star$};
					}{}
					\ifthenelse{\equal{#1}{d}}{
						\coordinate (A) at #2;
						\foreach \r in {0.2,0.5,...,\xboundtail}
							\draw[\bondcol] (A)+(0,-\r) node {\tiny $\star$};
					}{}
					\fill[\nodecol] #2 circle (\noderadius);
					\draw[\bondcol] #2 circle (\noderadius);
				}
								
				\draw[dotted] (-\db,-\db) rectangle ({\volsize+\db},{\volsize+\db});
				\clip (-\db,-\db) rectangle ({\volsize+\db},{\volsize+\db});

				\draw[gray,very thin] (-1,-1) grid ({\volsize+1},{\volsize+1});

				\draw[very thick] (2,3) -- (2,5);
				\draw[very thick] (3,2) -- (3,5);
				\draw[very thick] (4,2) -- (4,5);
				\draw[very thick] (5,2) -- (5,6);
				\draw[very thick] (6,4) -- (6,6);
				\draw[very thick] (3,2) -- (6,2);
				\draw[very thick] (2,3) -- (5,3);
				\draw[very thick] (2,4) -- (6,4);
				\draw[very thick] (2,5) -- (3,5);
				\draw[very thick] (4,5) -- (6,5);
				\draw[very thick] (5,6) -- (6,6);
				
				\xbound[l]{(2,3)};
				\xbound[l]{(2,4)};
				\xbound[l]{(2,5)};
				\xbound[u]{(2,5)};
				\xbound[u]{(3,5)};
				\xbound[r]{(3,5)};
				\xbound[l]{(4,5)};
				\xbound[u]{(4,5)};
				\xbound[l]{(5,6)};
				\xbound[u]{(5,6)};
				\xbound[u]{(6,6)};
				\xbound[r]{(6,6)};
				\xbound[r]{(6,5)};
				\xbound[r]{(6,4)};
				\xbound[d]{(6,4)};
				\xbound[r]{(5,3)};
				\xbound[u]{(6,2)};
				\xbound[r]{(6,2)};
				\xbound[d]{(6,2)};
				\xbound[d]{(5,2)};
				\xbound[d]{(4,2)};
				\xbound[d]{(3,2)};
				\xbound[l]{(3,2)};
				\xbound[d]{(2,3)};
				
			\end{tikzpicture}
		
		&

			\begin{tikzpicture}[scale=.4,line cap=round]
				\newdimen\noderadius
				\noderadius=1.5mm
				\def\nodecol{white}
				\def\bondcol{black}
				
				\def\volsize{8}
				\def\db{0.25}

				\def\xboundtail{0.9}
				\def\xbound[#1]#2{%
					\ifthenelse{\equal{#1}{l}}{
						\coordinate (A) at #2;
						\foreach \r in {0.2,0.5,...,\xboundtail}
							\draw[\bondcol] (A)+(-\r,0) node {\tiny $\star$};
					}{}
					\ifthenelse{\equal{#1}{u}}{
						\coordinate (A) at #2;
						\foreach \r in {0.2,0.5,...,\xboundtail}
							\draw[\bondcol] (A)+(0,\r) node {\tiny $\star$};
					}{}
					\ifthenelse{\equal{#1}{r}}{
						\coordinate (A) at #2;
						\foreach \r in {0.2,0.5,...,\xboundtail}
							\draw[\bondcol] (A)+(\r,0) node {\tiny $\star$};
					}{}
					\ifthenelse{\equal{#1}{d}}{
						\coordinate (A) at #2;
						\foreach \r in {0.2,0.5,...,\xboundtail}
							\draw[\bondcol] (A)+(0,-\r) node {\tiny $\star$};
					}{}
					\fill[\nodecol] #2 circle (\noderadius);
					\draw[\bondcol] #2 circle (\noderadius);
				}
				
				\draw[dotted] (-\db,-\db) rectangle ({\volsize+\db},{\volsize+\db});
				\clip (-\db,-\db) rectangle ({\volsize+\db},{\volsize+\db});

				\draw[gray,very thin] (-1,-1) grid ({\volsize+1},{\volsize+1});

				\draw[very thick] (-1,0) -- (9,0);
				\draw[very thick] (-1,1) -- (9,1);
				\draw[very thick] (-1,2) -- (2,2);
				\draw[very thick] (7,2) -- (8,2);
				\draw[very thick] (-1,3) -- (1,3);
				\draw[very thick] (6,3) -- (8,3);
				\draw[very thick] (-1,4) -- (1,4);
				\draw[very thick] (7,4) -- (8,4);
				\draw[very thick] (-1,5) -- (1,5);
				\draw[very thick] (7,5) -- (8,5);
				\draw[very thick] (-1,6) -- (4,6);
				\draw[very thick] (7,6) -- (8,6);
				\draw[very thick] (-1,7) -- (9,7);
				\draw[very thick] (-1,8) -- (9,8);
				
				\draw[very thick] (0,-1) -- (0,9);
				\draw[very thick] (1,-1) -- (1,9);
				\draw[very thick] (2,-1) -- (2,2);
				\draw[very thick] (2,6) -- (2,9);
				\draw[very thick] (3,-1) -- (3,1);
				\draw[very thick] (3,6) -- (3,9);
				\draw[very thick] (4,-1) -- (4,1);
				\draw[very thick] (4,6) -- (4,9);
				\draw[very thick] (5,-1) -- (5,1);
				\draw[very thick] (5,7) -- (5,9);
				\draw[very thick] (6,-1) -- (6,1);
				\draw[very thick] (6,7) -- (6,9);
				\draw[very thick] (7,-1) -- (7,9);
				\draw[very thick] (8,-1) -- (8,9);				
				
				\xbound[r]{(2,2)};
				\xbound[u]{(2,2)};
				\xbound[r]{(1,3)};
				\xbound[r]{(1,3)};
				\xbound[r]{(1,4)};
				\xbound[r]{(1,5)};
				\xbound[d]{(2,6)};
				\xbound[d]{(3,6)};
				\xbound[d]{(4,6)};
				\xbound[r]{(4,6)};
				\xbound[d]{(5,7)};
				\xbound[d]{(6,7)};
				\xbound[l]{(7,6)};
				\xbound[l]{(7,5)};
				\xbound[l]{(7,4)};
				\xbound[u]{(6,3)};
				\xbound[l]{(6,3)};
				\xbound[d]{(6,3)};
				\xbound[l]{(7,2)};
				\xbound[u]{(6,1)};
				\xbound[u]{(6,1)};
				\xbound[u]{(5,1)};
				\xbound[u]{(4,1)};
				\xbound[u]{(3,1)};
				
			\end{tikzpicture}

		\medskip \\
		
		(a) & (b)
		
		\bigskip \\
		
		\begin{tikzpicture}[scale=.4,line cap=round]
				\def\volsize{8}
				\def\db{0.25}
								
				\draw[dotted] (-\db,-\db) rectangle ({\volsize+\db},{\volsize+\db});
				\clip (-\db,-\db) rectangle ({\volsize+\db},{\volsize+\db});

				\draw[gray,very thin] (-1,-1) grid ({\volsize+1},{\volsize+1});

				\draw[very thick] (2,3) -- (2,5);
				\draw[very thick] (3,2) -- (3,5);
				\draw[very thick] (4,2) -- (4,5);
				\draw[very thick] (5,2) -- (5,6);
				\draw[very thick] (6,4) -- (6,6);
				\draw[very thick] (3,2) -- (6,2);
				\draw[very thick] (2,3) -- (5,3);
				\draw[very thick] (2,4) -- (6,4);
				\draw[very thick] (2,5) -- (3,5);
				\draw[very thick] (4,5) -- (6,5);
				\draw[very thick] (5,6) -- (6,6);

				\def\dist{0.27}
				\draw[dashed] (2-\dist,3-\dist) --
					(2-\dist,5+\dist) --
					(3+\dist,5+\dist) --
					(3+\dist,4+\dist) --
					(4-\dist,4+\dist) --
					(4-\dist,5+\dist) --
					(5-\dist,5+\dist) --
					(5-\dist,6+\dist) --
					(6+\dist,6+\dist) --
					(6+\dist,4-\dist) --
					(5+\dist,4-\dist) --
					(5+\dist,2+\dist) --
					(6+\dist,2+\dist) --
					(6+\dist,2-\dist) --
					(3-\dist,2-\dist) --
					(3-\dist,3-\dist) --
					cycle;
				
			\end{tikzpicture}
		
		&

			\begin{tikzpicture}[scale=.4,line cap=round]
				\def\volsize{8}
				\def\db{0.25}
				
				\draw[dotted] (-\db,-\db) rectangle ({\volsize+\db},{\volsize+\db});
				\clip (-\db,-\db) rectangle ({\volsize+\db},{\volsize+\db});

				\draw[gray,very thin] (-1,-1) grid ({\volsize+1},{\volsize+1});

				\draw[very thick] (-1,0) -- (9,0);
				\draw[very thick] (-1,1) -- (9,1);
				\draw[very thick] (-1,2) -- (2,2);
				\draw[very thick] (7,2) -- (8,2);
				\draw[very thick] (-1,3) -- (1,3);
				\draw[very thick] (6,3) -- (8,3);
				\draw[very thick] (-1,4) -- (1,4);
				\draw[very thick] (7,4) -- (8,4);
				\draw[very thick] (-1,5) -- (1,5);
				\draw[very thick] (7,5) -- (8,5);
				\draw[very thick] (-1,6) -- (4,6);
				\draw[very thick] (7,6) -- (8,6);
				\draw[very thick] (-1,7) -- (9,7);
				\draw[very thick] (-1,8) -- (9,8);
				
				\draw[very thick] (0,-1) -- (0,9);
				\draw[very thick] (1,-1) -- (1,9);
				\draw[very thick] (2,-1) -- (2,2);
				\draw[very thick] (2,6) -- (2,9);
				\draw[very thick] (3,-1) -- (3,1);
				\draw[very thick] (3,6) -- (3,9);
				\draw[very thick] (4,-1) -- (4,1);
				\draw[very thick] (4,6) -- (4,9);
				\draw[very thick] (5,-1) -- (5,1);
				\draw[very thick] (5,7) -- (5,9);
				\draw[very thick] (6,-1) -- (6,1);
				\draw[very thick] (6,7) -- (6,9);
				\draw[very thick] (7,-1) -- (7,9);
				\draw[very thick] (8,-1) -- (8,9);				
				
				\def\dist{0.27}
				\draw[dashed] (1+\dist,2+\dist) --
					(1+\dist,6-\dist) --
					(4+\dist,6-\dist) --
					(4+\dist,7-\dist) --
					(7-\dist,7-\dist) --
					(7-\dist,3+\dist) --
					(6-\dist,3+\dist) --
					(6-\dist,3-\dist) --
					(7-\dist,3-\dist) --
					(7-\dist,1+\dist) --
					(2+\dist,1+\dist) --
					(2+\dist,2+\dist) --
					cycle;
				
			\end{tikzpicture}

		\medskip \\
		
		(a') & (b')
		
		\end{tabular}
		\end{center}
	\caption{%
		(a) A disorder contour and its defining configuration.
		(b) An order contour and its defining configuration.
		(a') and (b') Geometric illustrations of (a) and (b).
	}
	\label{fig:contour:int-ext}
\end{figure}
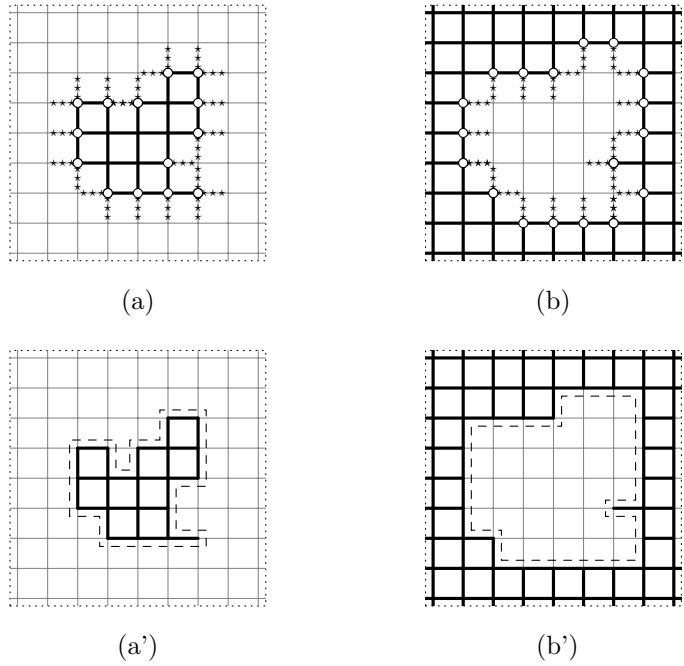

Two contours are said to be \emph{mutually compatible}
if they are disjoint (as subsets of $\xS\times \xB$).
Let us emphasize that two mutually compatible contours are allowed
to share either sites or bonds, but not pairs. 

If $X\subseteq \xB$ is an arbitrary configuration,
there could be several ways to partition its boundary $\xpartial X$
into mutually compatible contours. 
One way to construct such decomposition in an unambiguous way
is as follows: first, we partition $(S(X),X)$ into
its maximal connected components $(S(C_i),C_i)$.
Then $\xpartial C_i$ form a partitioning of $\xpartial X$.
Now, the maximal co-connected components of every $C_i$
are contours that we identify as
the \emph{contours of} $X$.

The above decomposition allows us to think of $\xpartial X$
as a family of mutually compatible contours,
which we call the \emph{contour family} of $X$.
Let us recall that the contour family of a configuration $X$
uniquely determines $X$.
However, note that not every family of mutually compatible
contours corresponds to a configuration.
In particular, in a contour family of a configuration~$X$,
between every two nested finite contours of the same type,
there necessarily lies a contour of the other type.
This requirement induces a long-range constraint among contours,
which raises some difficulties in dealing with the contours.
We will see later how to get rid of such a constraint.
Let us call a family $\partial$ of contours \emph{admissible}
if it is the contour family of a configuration $X\subseteq \xB$.
We shall denote by $\Delta$ the set of all
admissible contour families.
A contour $\gamma$ in a mutually compatible family $\partial$
of contours is said to be \emph{external}
if it is not in the interior of any other contour in $\partial$.
Note that if $\partial$ is an admissible contour family
with no infinite contours,
all the external contours in $\partial$ are necessarily
of the same type.


Having formalized the notions of boundary and contours,
we can now express
the weight of a configuration of the $r$-biased random-cluster model
in terms of the ``energy'' of its ordered and disordered regions
and the correction term due to the contours separating them.
The one-to-one correspondence between the configurations
and the admissible families of contours allows us to
write the partition functions as a sum over contour families.
The ordered/disordered boundary conditions on
the configurations translate into the constraints for
the corresponding contour family that
the outermost contours in the volume be of the order/disorder type.

Let $\Lambda$ be a \emph{volume} in the lattice, by which, from now on,
we shall mean a finite subgraph of $\xL$ without ``holes''.
More precisely, we assume that if we remove the subgraph $\Lambda$ from $\xL$,
the remaining subgraph is connected.
Let us denote by $\Delta_\Lambda^\disorder$
the set of all
admissible contour families whose contours are in $\Lambda$
(i.e., their bonds are chosen from the bonds of $\Lambda$)
and whose external contours are all of the disorder type.
Similarly, let $\Delta_\Lambda^\order$ denote
the set of admissible contour families in $\Lambda$
whose external contours are all of the order type.
The partition function for the $r$-biased random-cluster model
in a volume $\Lambda$ with \emph{disordered}
(resp., \emph{ordered}) \emph{boundary conditions}
can be defined as
\begin{align}
\label{eq:rc-disordered-def:old}
	Z^{\RC.\disorder}_{p_\beta,q,r}(\Lambda) &=
		(q+r)^{\frac{\abs{\xpartial B(\Lambda)}}{4}}
		\sum_{\partial\in \Delta_\Lambda^\disorder}
		\xe^{
			-\abs{R_\Lambda^\xorder(\partial)}\cdot e(\xB)
			-\abs{R_\Lambda^\xdisorder(\partial)}\cdot e(\varnothing)
		}
		\prod_{\gamma\in\partial}\rho(\gamma) \;, \\
\label{eq:rc-ordered-def:old}
	Z^{\RC.\order}_{p_\beta,q,r}(\Lambda) &=
		q
		\sum_{\partial\in \Delta_\Lambda^\order}
		\xe^{
			-\abs{R_\Lambda^\xorder(\partial)}\cdot e(\xB)
			-\abs{R_\Lambda^\xdisorder(\partial)}\cdot e(\varnothing)
		}
		\prod_{\gamma\in\partial}\rho(\gamma) \;,
\end{align}
where $\rho(\gamma)$ is the \emph{weight} of a contour $\gamma$
and is given by
\begin{align}
	\rho(\gamma) &\IsDef \begin{cases}
		(q+r)^{-\frac{1}{4}\abs{\gamma}}\;,
			&\text{if $\gamma$ order,} \\
		q\,(q+r)^{-\frac{1}{4}\abs{\gamma}} \;,
			&\text{if $\gamma$ disorder,}
	\end{cases}		
\end{align}
and $R_\Lambda^\xorder(\partial)$ and $R_\Lambda^\xdisorder(\partial)$
denote, respectively, the sets of ordered and disordered bonds in $\Lambda$
of the configuration corresponding to $\partial$.

The above definitions are consistent with
the definitions given in~(\ref{eq:rc-partition-function:disordered})
and~(\ref{eq:rc-partition-function:ordered})
when $\Lambda=\Lambda_{n+1}$ is a square.
Namely, for $\Lambda=\Lambda_{n+1}$,
if $X$ is the corresponding configuration
of a family $\partial\in\Delta_{\Lambda_{n+1}}^\disorder$,
the restriction of $X$ to $B(\Lambda_{n+1})$
is an element of $\mathcal{X}_{\Lambda_{n+1}}^\disorder$.
Conversely, every element of $\mathcal{X}_{\Lambda_{n+1}}^\disorder$
has a unique infinite-volume extension whose corresponding
contour family is in $\Delta_{\Lambda_{n+1}}^\disorder$.
A similar correspondence holds between
$\mathcal{X}_{\Lambda_{n+1}}^\order$ and $\Delta_{\Lambda_{n+1}}^\order$.
For the proof of the equivalence of the two definitions
see Appendix~\ref{apx:contour}.

We emphasize that the factors $(q+r)^{\frac{\abs{\xpartial B(\Lambda)}}{4}}$
and $q$ in front of the partition functions~(\ref{eq:rc-disordered-def:old})
and~(\ref{eq:rc-ordered-def:old}) do not contribute to
the pressure function~$f^\RC(\beta)$:
in the thermodynamic limit, they are swallowed by the size of the volume.
Hence, to avoid heavy notation
--- with all due apologies to the reader ---
we \emph{re-define} the partition functions of the $r$-biased random-cluster model
with disordered/ordered boundary conditions as
\begin{align}
\label{eq:rc-disordered-def}
	Z^{\RC.\disorder}(\Lambda) &=
		\sum_{\partial\in \Delta_\Lambda^\disorder}
		\xe^{
			-\abs{R_\Lambda^\xorder(\partial)}\cdot e(\xB)
			-\abs{R_\Lambda^\xdisorder(\partial)}\cdot e(\varnothing)
		}
		\prod_{\gamma\in\partial}\rho(\gamma) \;, \\
\label{eq:rc-ordered-def}
	Z^{\RC.\order}(\Lambda) &=
		\sum_{\partial\in \Delta_\Lambda^\order}
		\xe^{
			-\abs{R_\Lambda^\xorder(\partial)}\cdot e(\xB)
			-\abs{R_\Lambda^\xdisorder(\partial)}\cdot e(\varnothing)
		}
		\prod_{\gamma\in\partial}\rho(\gamma) \;.
\end{align}
From now on, every time we talk about the partition function
of the $r$-biased random-cluster model, we will be referring to
the latter definitions.

As was mentioned in the introduction,
we would like to express the two partition functions
in terms of two (standard) contour models.
The purpose of this is to make use of
the machinery available for contour models; namely,
a result providing an estimate on the convergence of
the corresponding free energy functions (Proposition~\ref{prop:damped}),
and the Peierls estimate for the probability of
the appearance of a contour.
The main features of the contour models
that are required in the above tools are (see~\cite{Kot94})
\begin{enumerate}[ i)]
	\item independence, and
	\item damping.
\end{enumerate}
Unfortunately, the contours of the contour representation
of the $r$-biased random-cluster partition functions are
not independent (due to the long-range constraint).
In the following section, we will see how to
achieve the independence among contours,
by rewriting the partition functions
in terms of abstract contour models.
As in the standard random-cluster model (see~\cite{LaaMesMirRuiShl91}),
we need two different such contour models,
one for each of the two boundary conditions.


\subsection{Contour Models}
\label{sec:contour:abstract}
In this section, we want to resolve the issue of
long-range constraints between contours.
Recall that the admissibility condition requires
the contours of a family to be alternating between disorder and order contours,
and this imposes a long-range constraint between contours.
For example, two nested contours of the disorder type
(no matter how far from each other) are ``aware''
of the presence of an order contour separating them.
As a result, if we remove a contour
from an admissible family, the admissibility could be lost.

In order to get rid of this constraint,
we use two abstract contour models
in which the contours are all of the same type
and the admissibility condition is replaced by
mere mutual compatibility.
The weights of the contours in each of the abstract models
will be chosen in such a way to guarantee
that the ensuing partition functions are equal (up to a factor)
to each of the partition functions for the $r$-biased
random-cluster model.

A \emph{contour model} is specified by a function
$\chi:\Gamma\to\xR$, assigning a weight $\chi(\gamma)$
to each contour $\gamma\in\Gamma$.
The configurations of the model are
families of mutually compatible (i.e., disjoint) contours in $\xL$.
Let us denote the set of all such families
by $\mathcal{M}$, and the set of all elements of $\mathcal{M}$
whose contours are in a volume $\Lambda$ by $\mathcal{M}_\Lambda$.
The partition function of the model in $\Lambda$
is given by
\begin{align}
	\mathscr{Z}(\Lambda\,|\,\chi) &=
		\sum_{\partial\in\mathcal{M}_\Lambda}\prod_{\gamma\in\partial}
			\chi(\gamma) \;.
\end{align}
In the following lemma, we will see how to represent
the partition functions of the $r$-biased random-cluster model
with disordered and ordered boundary conditions, each in terms of
of the partition function of a contour model,
with a particular choice of the weight function.
In fact, the contour model associated to the disordered boundary condition
will not involve order contours.  This is reflected by
the fact that in this model each order contour has weight zero.
Similarly, the contour model for the ordered boundary condition
involves only order contours.

To set the stage for the lemma, we rewrite the partition functions
$Z^{\RC.\disorder}(\Lambda)$ and $Z^{\RC.\order}(\Lambda)$
in a form resembling more the contour model
partition function
$\mathscr{Z}(\Lambda\,|\,\chi)$.
That is,
\begin{align}
	Z^{\RC.\disorder}(\Lambda) &=
		\xe^{-\abs{B(\Lambda)}\cdot e(\varnothing)}
		\sum_{\partial\in\Delta_\Lambda^\disorder}
		\prod_{\gamma\in\partial}
		\widetilde{\rho}(\gamma) \;, \\
	Z^{\RC.\order}(\Lambda) &=
		\xe^{-\abs{B(\Lambda)}\cdot e(\xB)}
		\sum_{\partial\in\Delta_\Lambda^\order}
		\prod_{\gamma\in\partial}
		\widetilde{\rho}(\gamma) \;,
\end{align}
where
\begin{align}
	\widetilde{\rho}(\gamma) &= \begin{cases}
		\rho(\gamma)\cdot\xe^{
			-\abs{B(\interior\gamma)}\cdot
			\left(e(\xB)-e(\varnothing)\right)
		}\;,\quad & \text{if $\gamma$ is disorder,} \\
		\rho(\gamma)\cdot\xe^{
			-\abs{B(V(\gamma))}\cdot
			\left(e(\varnothing)-e(\xB)\right)
		}\;, & \text{if $\gamma$ is order.}
	\end{cases}
\end{align}
Let us recall that the set $\Delta_\Lambda^\disorder$
(resp., $\Delta_\Lambda^\order$)
does not contain only families of disorder (resp., order) contours,
but all families compatible with the disordered (resp., ordered) boundary condition.
To make the proof more transparent, let us define
\begin{align}
\label{eq:Z-tilde-f}
	Y^\xdisorder(\Lambda) &=
		\sum_{\partial\in\Delta_\Lambda^\disorder}
		\prod_{\gamma\in\partial}
		\widetilde{\rho}(\gamma) \;, \\
\label{eq:Z-tilde-w}
	Y^\xorder(\Lambda) &=
		\sum_{\partial\in\Delta_\Lambda^\order}
		\prod_{\gamma\in\partial}
		\widetilde{\rho}(\gamma) \;,
\end{align}
so that
\begin{align}
	Z^{\RC.\disorder}(\Lambda) &=
		\xe^{-\abs{B(\Lambda)}e(\varnothing)}\cdot
		Y^\xdisorder(\Lambda) \;, \\
	Z^{\RC.\order}(\Lambda) &=
		\xe^{-\abs{B(\Lambda)}e(\xB)}\cdot
		Y^\xorder(\Lambda) \;.
\end{align}
Notice that the above contour representation for the partition functions
$Z^{\RC.\disorder}(\Lambda)$ and $Z^{\RC.\order}(\Lambda)$
lacks the condition of independence between compatible contours.

The following lemma is similar to Lemma~1 of~\cite{Kot94}.
\begin{lemma}
\label{lem:contour-model}
	The partition functions for the $r$-biased random-cluster model
	on volume $\Lambda$ with the disordered and ordered boundary conditions
	can be written as
	\begin{align}
		Z^{\RC.\disorder}(\Lambda) &=
			\xe^{-\abs{B(\Lambda)}e(\varnothing)}
			\mathscr{Z}(\Lambda\,|\,\xi^\xdisorder) \;, \\
		Z^{\RC.\order}(\Lambda) &=
			\xe^{-\abs{B(\Lambda)}e(\xB)}
			\mathscr{Z}(\Lambda\,|\,\xi^\xorder) \;,
	\end{align}
	where the weights $\xi^\xdisorder$ and $\xi^\xorder$
	are defined by
	\begin{align}
		\xi^\xdisorder(\gamma) &=
			\begin{cases}
				\displaystyle{%
				\rho(\gamma)\frac{
					Z^{\RC.\order}(\interior\gamma)
				}{
					Z^{\RC.\disorder}(\interior\gamma)
				}\;,
				}\qquad	& \text{if $\gamma$ disorder,} \\
				0\;, 	& \text{otherwise,}
			\end{cases}
	\end{align}
	and
	\begin{align}
	\label{eq:xi-order}
		\xi^\xorder(\gamma) &=
			\begin{cases}
				\displaystyle{%
				\rho(\gamma)\xe^{\abs{B(\gamma)}e(\xB)}\frac{
					Z^{\RC.\disorder}(V(\gamma))
				}{
					Z^{\RC.\order}(\interior\gamma)
				}\;,
				}\qquad	& \text{if $\gamma$ is order,} \\
				0\;,		& \text{otherwise.}
			\end{cases}
	\end{align}
\end{lemma}
\begin{proof}
	The key step to prove the lemma is to write a recursion
	for the above partition functions by factoring
	the contribution of the interior of each external contour.
	Let us denote by $\mathcal{E}^\disorder_\Lambda$
	the set of mutually compatible families of disorder contours whose elements
	are all external.  (We include the empty family in $\mathcal{E}^\disorder_\Lambda$.)
	Note that the elements of $\mathcal{E}^\disorder_\Lambda$
	are all admissible and in $\Delta^\disorder_\Lambda$.
	Moreover, for each admissible family $\partial\in\Delta^\disorder_\Lambda$,
	the sub-family of $\partial$ consisting of its external contours
	is in $\mathcal{E}^\disorder_\Lambda$.
	Similarly,  we denote by $\mathcal{E}^\order_\Lambda$
	the set of mutually compatible families of order contours whose elements
	are all external.
	The partition functions $Y^\xdisorder$ and $Y^\xorder$
	satisfy the following recursions:
	\begin{align}
		\label{eq:recursion:disordered}
		Y^\xdisorder(\Lambda) &=
			\sum_{\theta\in\mathcal{E}^\disorder_\Lambda}
			\prod_{\gamma\in\theta}\,
			\widetilde{\rho}(\gamma)\cdot Y^\xorder(\interior\gamma)
				\;, \\
		Y^\xorder(\Lambda) &=
			\sum_{\theta\in\mathcal{E}^\order_\Lambda}
			\prod_{\gamma\in\theta}\,
			\widetilde{\rho}(\gamma)\cdot Y^\xdisorder(V(\gamma)) \;.
	\end{align}
	Similar recursions hold for the contour model partition functions
	$\mathscr{Z}(\cdot\,|\,\xi^\xdisorder)$ and
	$\mathscr{Z}(\cdot\,|\,\xi^\xorder)$:
	\begin{align}
		\label{eq:recursion:disordered:abstract}
		\mathscr{Z}(\Lambda\,|\,\xi^\xdisorder) &=
			\sum_{\theta\in\mathcal{E}^\disorder_\Lambda}
			\prod_{\gamma\in\theta}\,
			\xi^\xdisorder(\gamma)\cdot\mathscr{Z}(\interior\gamma\,|\,\xi^\xdisorder) \;, \\
		\mathscr{Z}(\Lambda\,|\,\xi^\xorder) &=
			\sum_{\theta\in\mathcal{E}^\order_\Lambda}
			\prod_{\gamma\in\theta}\,
			\xi^\xorder(\gamma)\cdot\mathscr{Z}(\interior\gamma\,|\,\xi^\xorder) \;.
	\end{align}
	Note that, since every order contour is weighted $0$ by $\xi^\xdisorder$,
	we can ignore in $\mathscr{Z}(\cdot\,|\,\xi^\xdisorder)$
	the families containing order contours,
	and similarly the disorder contours can be ignored
	in $\mathscr{Z}(\cdot\,|\,\xi^\xorder)$.
	
	We use induction on volume $\Lambda$ to prove that
	$Y^\xdisorder(\Lambda)=\mathscr{Z}(\Lambda\,|\,\xi^\xdisorder)$.
	Suppose that for every sub-volume $\Lambda'\subsetneq\Lambda$
	we have
	$Y^\xdisorder(\Lambda')=
		\mathscr{Z}(\Lambda'\,|\,\xi^\xdisorder)$.
	Let $\theta\in\mathcal{E}_{\Lambda}^\disorder$.
	We want to show that the terms corresponding to $\theta$
	in the recursion
	formulas~(\ref{eq:recursion:disordered}) and~(\ref{eq:recursion:disordered:abstract})
	for $Y^\xdisorder(\Lambda)$ and
	$\mathscr{Z}(\Lambda\,|\,\xi^\xdisorder)$
	are equal.
	If $\theta$ is empty, the equality is trivial
	(we consider the product over an empty set to be $1$).
	Otherwise, for every $\gamma\in\theta$,
	we have $\interior\gamma\subsetneq\Lambda$,
	which implies
	$\mathscr{Z}(\interior\gamma\,|\,\xi^\xdisorder) =
		Y^\xdisorder(\interior\gamma)$.
	Using the definitions of
	$\widetilde{\rho}$ and $\xi^\xdisorder$
	we obtain that
	\begin{align}
		\prod_{\gamma\in\theta}\,
			\widetilde{\rho}(\gamma)\cdot Y^\xorder(\interior\gamma)
		&=
		\prod_{\gamma\in\theta}\,
			\xi^\xdisorder(\gamma)\cdot Y^\xdisorder(\interior\gamma) \;.
	\end{align}
	Therefore, $Y^\xdisorder(\Lambda)=
	\mathscr{Z}(\Lambda\,|\,\xi^\xdisorder)$.
	The starting point of the induction is
	when the only element of $\mathcal{E}^\disorder_\Lambda$
	is the empty family.
	
	The argument for
	$Y^\xorder(\Lambda)=
		\mathscr{Z}(\Lambda\,|\,\xi^\xorder)$
	is similar.
\end{proof}

Note that there is no complete correspondence
between the configurations of the $r$-biased random-cluster model
and the contour families of the corresponding
abstract contour model.
Nevertheless, the probability of appearance of a contour
as an external contour is the same in both models.
Let
$\phi_\Lambda^\order$ 
denote the probability distribution
associated to $Z^{\RC.\order}(\Lambda)$. 
We consider $\phi_\Lambda^\order$ as a measure on
the infinite-volume bond configurations $X\subseteq\xB$,
which is concentrated on the set $\{X: \xpartial X\in\Delta_\Lambda^\order\}$.
Likewise, $\phi_\Lambda^\disorder$ will denote the
measure corresponding to $Z^{\RC.\disorder}(\Lambda)$,
which is concentrated on the set $\{X: \xpartial X\in\Delta_\Lambda^\disorder\}$.

\begin{corollary}
	Let $\Lambda$ be a finite volume
	and $\theta\in\mathcal{E}_\Lambda^\order$
	a family of external mutually compatible order contours.
	Then,
	\begin{align}
		\phi_\Lambda^\order\{X: \xpartial_{\extern}X=\theta\}
			&=
			\frac{%
				\prod_{\gamma\in\theta}
				\xi^\xorder(\gamma)\mathscr{Z}(\interior\gamma\,|\,\xi^\xorder)
			}{%
				\mathscr{Z}(\Lambda\,|\,\xi^\xorder)
			} \;,
	\end{align}
	where
	$\xpartial_{\extern}X$ is the family of external contours
	of $X$.
	A similar statement holds for
	the probability of 
	families of external mutually compatible disorder contours
	under 
	$\phi_\Lambda^\disorder$.
\end{corollary}
\begin{proof}
	\begin{align}
		\phi_\Lambda^\order\{X: \xpartial_{\extern}X=\theta\} &=
			\frac{%
				\prod_{\gamma\in\theta}\widetilde{\rho}(\gamma)
				Y^\xdisorder(V(\gamma))
			}{%
				Y^\xorder(\Lambda)
			} \\
		&=
			\frac{%
				\prod_{\gamma\in\theta}
				\left[\rho(\gamma)\cdot
				\xe^{
					-\abs{B(V(\gamma))}\cdot
					\left(e(\varnothing-e(\xB))\right)
				}
				\frac{%
					Y^\xdisorder(V(\gamma))
				}{%
					Y^\xorder(\interior\gamma)
				}
				\right]
				Y^\xorder(\interior\gamma)
			}{%
				Y^\xorder(\Lambda)
			} \\
		&=
			\frac{%
				\prod_{\gamma\in\theta}
				\xi^\xorder(\gamma)\mathscr{Z}(\interior\gamma\,|\,\xi^\xorder)
			}{%
				\mathscr{Z}(\Lambda\,|\,\xi^\xorder)
			}
	\end{align}
\end{proof}

The next corollary provides an estimate
for the probability that a finite region of the lattice 
is surrounded by a contour (see e.g.~\cite{Gri06}).
For a finite set of sites $A$ in the lattice,
let $\Gamma_A$ denote the set of all finite contours
that have $A$ in their interiors.
\begin{corollary}
	For every finite volume $\Lambda$ and
	every finite set $A\subseteq S(\Lambda)$ we have
	\begin{align}
		\phi_\Lambda^\order\{
			X: \xpartial X\cap\Gamma_A\neq\varnothing
		\}
		&\leq
		\sum_{\gamma\in\Gamma_A} \xi^\xorder(\gamma)
		=
		\sum_{\substack{\gamma\in\Gamma_A\\ \gamma\text{ \rm order}}}
		\xi^\xorder(\gamma) \;.
	\end{align}
	A similar bound holds in the disordered case.
\end{corollary}
\begin{proof}
	Taking into account the ordered boundary condition,
	we have that
	if $A$ is surrounded by a contour in $\Lambda$, it is also
	surrounded by an external order contour in $\Lambda$, that is,
	\begin{align}
		\{X: \xpartial X\cap\Gamma_A\neq\varnothing\} &=
		\{X: \xpartial_\extern X\cap\Gamma_A\neq\varnothing\} \;.
	\end{align}
	By the previous corollary, we can bound the
	probability of a contour $\gamma$ appearing as
	an external contour by
	\begin{align}
		\phi_\Lambda^\order \{X: \xpartial_\extern X\ni\gamma\} &=
			\sum_{
				\substack{\theta\in\mathcal{E}^\order_\Lambda\\
				\theta\ni\gamma
			}}
			\frac{%
				\prod_{\widehat{\gamma}\in\theta}
				\xi^\xorder(\widehat{\gamma})
				\mathscr{Z}(\interior\widehat{\gamma}\,|\,\xi^\xorder)
			}{%
				\mathscr{Z}(\Lambda\,|\,\xi^\xorder)
			} \\
		&=
			\frac{%
				\xi^\xorder(\gamma)
					\mathscr{Z}(\interior\gamma\,|\,\xi^\xorder)
				\sum_{
					\substack{\theta\in\mathcal{E}^\order_\Lambda\\
					\theta\ni\gamma}
				} \prod_{\substack{\widehat{\gamma}\in\theta\\ \widehat{\gamma}\neq\gamma}}
				\xi^\xorder(\widehat{\gamma})
				\mathscr{Z}(\interior\widehat{\gamma}\,|\,\xi^\xorder)
			}{%
				\mathscr{Z}(\Lambda\,|\,\xi^\xorder)
			} \\
		&\leq
			\frac{%
				\xi^\xorder(\gamma)
				\mathscr{Z}(\interior\gamma\,|\,\xi^\xorder)
				\mathscr{Z}(\Lambda\setminus\interior\gamma\,|\,\xi^\xorder)	
			}{%
				\mathscr{Z}(\Lambda\,|\,\xi^\xorder)
			} \\
		&\leq
			\xi^\xorder(\gamma) \;.
	\end{align}
	The last step follows from the fact that all the terms
	in the partition function
	$\mathscr{Z}(\Lambda\,|\,\xi^\xorder)$ are non-negative, hence
	\begin{align}
		\mathscr{Z}(\Lambda\,|\,\xi^\xorder) &\geq
		\mathscr{Z}(\interior\gamma\,|\,\xi^\xorder)
		\mathscr{Z}(\Lambda\setminus\interior\gamma\,|\,\xi^\xorder) \;.
	\end{align}
	We obtain that
	\begin{align}
		\phi_\Lambda^\order\{
			X: \xpartial X\cap\Gamma_A\neq\varnothing
		\} &=
		\phi_\Lambda^\order\{
			X: \xpartial_\extern X\cap\Gamma_A\neq\varnothing
		\} \\
		&=
		\sum_{\gamma\in\Gamma_A}
		\phi_\Lambda^\order \{X: \xpartial_\extern X\ni\gamma\} \\
		&\leq
		\sum_{\gamma\in\Gamma_A} \xi^\xorder(\gamma) \;.
	\end{align}
\end{proof}


A standard argument using the positive correlation property of
$\phi_\Lambda^\order$ (resp., $\phi_\Lambda^\disorder$)
can be used to show that the thermodynamic limit of
$\phi_\Lambda^\order$ (resp., $\phi_\Lambda^\disorder$)
exists and is unique (see Appendix~\ref{apx:rc:properties}).
The limit measure $\phi^\order$ (resp., $\phi^\disorder$)
satisfies the same bound as in the above corollary.
If the weights $\xi^\xorder$ (resp., $\xi^\xdisorder$)
decay sufficiently fast, the latter bound
implies that under $\phi^\order$ (resp., $\phi^\disorder$),
the configuration of the model almost surely consists of
a unique infinite sea of order (resp., disorder)
with finite islands of disorder (resp., order).
By a ``sea'' of order (resp., disorder) in a random-cluster configuration
we mean a connected component of present (resp., absent) bonds.

\begin{corollary}
\label{cor:peierls}
	For every finite set $A\subseteq\xS$ we have
	\begin{align}
		\phi^\order\{
			X: \xpartial X\cap\Gamma_A\neq\varnothing
		\}
		&\leq
		\sum_{\gamma\in\Gamma_A} \xi^\xorder(\gamma)
		=
		\sum_{\substack{\gamma\in\Gamma_A\\ \gamma\text{ \rm order}}}
		\xi^\xorder(\gamma) \;.
	\end{align}
	Furthermore, if the sum $\sum_{\gamma\in\Gamma_A} \xi^\xorder(\gamma)$
	converges, we have
	\begin{align}
		\phi^\order\left(
			\begin{array}{c}
				\text{\rm $\exists$ unique infinite sea of order} \\
				\text{\rm with finite islands of disorder}
			\end{array}
		\right)=1 \;,
	\end{align}	
	A similar statement holds in the disordered case.
\end{corollary}
\begin{proof}
	As before, let $\Lambda_n$ denote the $(2n+1)\times(2n+1)$ central square in the lattice.
	For every $n$
	let us define $\Gamma_{A,\Lambda_n}$ as the set of all contours in $\Lambda_n$
	having $A$ in their interiors.  From the previous corollary, we know that
	for every $m>n$, the following bound holds:
	\begin{align}
		\phi_{\Lambda_m}^\order\{
			X: \xpartial X\cap\Gamma_{A,\Lambda_n}\neq\varnothing
		\}
		&\leq
		\sum_{\gamma\in\Gamma_A} \xi^\xorder(\gamma) \;.
	\end{align}
	Since the event $\{X: \xpartial X\cap\Gamma_{A,\Lambda_n}\neq\varnothing\}$
	is local, we obtain
	\begin{align}
		\phi^\order\{
			X: \xpartial X\cap\Gamma_{A,\Lambda_n}\neq\varnothing
		\}
		&\leq
		\sum_{\gamma\in\Gamma_A} \xi^\xorder(\gamma) \;.
	\end{align}	
	due to weak convergence of $\phi_{\Lambda_m}^\order$ to $\phi^\order$.
	Letting $n\to\infty$ the first claim follows.
	
	If $\sum_{\gamma\in\Gamma_A} \xi^\xorder(\gamma)$ converges,
	using a Borel-Cantelli argument, with probability~$1$,
	no infinite cascade of contours appears on the lattice.
	In particular, if we define
	\begin{align}
		\mathcal{S}^\xorder &\IsDef \left\{
			X: \begin{array}{c}
				\text{$X$ has a unique infinite sea of order} \\
				\text{with finite islands of disorder}
			\end{array}
		\right\} \;, \\
		\mathcal{S}^\xdisorder &\IsDef \left\{
			X: \begin{array}{c}
				\text{$X$ has a unique infinite sea of disorder} \\
				\text{with finite islands of order}
			\end{array}
		\right\} \;,
	\end{align}
	the latter implies that $\phi^\order(\mathcal{S}^\xorder\cup\mathcal{S}^\xdisorder)=1$.
	We show that, in fact, $\phi^\order(\mathcal{S}^\xdisorder)=0$.
	
	Let $A$ be a finite set of sites in the lattice.
	For every $X\in\mathcal{S}^\xdisorder$ one can find a volume $\Lambda$
	containing $A$ such that
	$\xpartial\left(X\cap B(\Lambda)\right)\in\Delta_\Lambda^\disorder$
	(i.e., the restriction of $X$ to $\Lambda$ is compatible with
	the disordered boundary condition).
	In particular, if we define
	\begin{align}
		\mathcal{C}_{A,\Lambda_n} &\IsDef \left\{
			X: \text{$\exists$ a finite volume $\Lambda\subseteq\Lambda_n$ with $S(\Lambda)\supseteq A$
			and $\xpartial\left(X\cap B(\Lambda)\right)\in \Delta_\Lambda^\disorder$}
		\right\} \;,
	\end{align}
	we have $\mathcal{C}_{A,\Lambda_1}\subseteq\mathcal{C}_{A,\Lambda_2}\subseteq\cdots$
	and $\mathcal{S}^\xdisorder\subseteq\bigcup_n\,\mathcal{C}_{A,\Lambda_n}$.
	
	If $m,n$ are integers with $m>n$, every configuration $X$ that is compatible with
	the ordered boundary condition on $\Lambda_m$ (i.e., $\xpartial X\in\Delta_{\Lambda_m}^\order$)
	and is in $\mathcal{C}_{A,\Lambda_n}$ necessarily has an order contour
	surrounding $A$.  Therefore, by the previous corollary, we have
	\begin{align}
		\phi^\order_{\Lambda_m}\left(\mathcal{C}_{A,\Lambda_n}\right) \leq
			\phi^\order_{\Lambda_m}\left\{X: \xpartial X\cap\Gamma_A\neq\varnothing\right\} \leq
			\sum_{\gamma\in\Gamma_A}\xi^\xorder(\gamma) \;.
	\end{align}
	Since $\mathcal{C}_{A,\Lambda_n}$ is a local event, by weak convergence
	of $\phi^\order_{\Lambda_m}$ to $\phi^\order$ we have
	\begin{align}
		\phi^\order\left(\mathcal{C}_{A,\Lambda_n}\right) \leq
			\sum_{\gamma\in\Gamma_A}\xi^\xorder(\gamma) \;.
	\end{align}
	Letting $n\to\infty$, we obtain
	\begin{align}
		\phi^\order(\mathcal{S}^\xdisorder) \leq
			\lim_{n\to\infty} \phi^\order(\mathcal{C}_{A,\Lambda_n}) \leq
			\sum_{\gamma\in\Gamma_A}\xi^\xorder(\gamma) \;.
	\end{align}
	The latter holds for every finite set $A\subseteq\xS$,
	which by the convergence of the series, implies that
	$\phi^\order(\mathcal{S}^\xdisorder)=0$.
\end{proof}


\section{Damping of Contour Weights}
\label{sec:damping}

One advantage of working with contour models is that
when the contour weights are sufficiently ``damped''
(i.e., decay exponentially in the length with a sufficiently fast rate)
the free energy exists and is bounded, and moreover,
the error in the finite-volume approximations of the free energy
is of the order of the size of the boundary of the finite volume.
This is the message of the following well-known proposition
(see e.g. Section~2 of~\cite{Zah84}, or Proposition~2.3 of~\cite{Sin82}).

\begin{proposition}
\label{prop:damped}
	Let $\tau>0$ be sufficiently large, and
	suppose that the weight function $\chi:\Gamma\to\xR$
	of a contour model satisfies
	$0\leq\chi(\gamma)\leq\xe^{-\tau\abs{\gamma}}$
	for every contour $\gamma$.
	Then, the limit
	\begin{align}
	\label{eq:pressure:contour-model}
		g(\chi)=\lim_{n\to\infty}
              \frac{1}{\abs{B(\Lambda_n)}}\log\mathscr{Z}(\Lambda_n\,|\,\chi)
	\end{align}
	exists and satisfies
	$g(\chi)\leq\sum_{\gamma:S(\gamma)\ni 0}\chi(\gamma)\leq\xe^{-\tau/2}$. 
	In particular $g(\chi)\rightarrow 0$ as $\tau$ tends to infinity. 
	
	Furthermore, there is a constant $C=C(\tau)$, such that $C\rightarrow 0$ 
	as $\tau$ goes to infinity, and for each finite volume
	$\Lambda\in\mathbb{Z}^d$
	\begin{align}
		\xe^{g(\chi)\abs{B(\Lambda)}-C(\tau)\abs{\partial\Lambda}}
		\leq \mathscr{Z}(\Lambda\,|\,\chi)
		\leq \xe^{g(\chi)\abs{B(\Lambda)}+C(\tau)\abs{\partial\Lambda}}\;,
	\end{align}
	where $\partial\Lambda$ denotes the \emph{boundary} of
	the volume $\Lambda$ and can be defined as the set of bonds
	that are not in $B(\Lambda)$ but are incident to $\Lambda$.
\end{proposition}	

The main purpose of this section is to identify conditions on
the parameters $(q+r)$ and $\beta$ under which the weights
$\xi^\xdisorder$ and $\xi^\xorder$ are damped
(i.e., satisfy the condition of the above proposition).
We will see that when $(q+r)$ is large, for any value of $\beta>0$
at least one of $\xi^\xdisorder$ and $\xi^\xorder$ is damped,
and moreover there exists a unique $\beta$ at which
both $\xi^\xdisorder$ and $\xi^\xorder$ are damped.
Let us remark that for sufficiently damped weights,
the sum appearing in Corollary~\ref{cor:peierls}
converges, implying that the corresponding phase is stable.

For $\tau>0$ large enough,
let us introduce the \emph{truncated} (i.e., artificially damped) weights
\begin{align}
\label{eq:weight:truncated}
	\bar{\xi}^\xorder(\gamma) &= \begin{cases}
			\xi^\xorder(\gamma)\;,\quad &
				\text{if $\xi^\xorder(\gamma)\leq \xe^{-\tau|\gamma|}$,} \\
			0\;, & \text{otherwise,}
		\end{cases}
\end{align}
and similarly for $\bar\xi^\xdisorder(\gamma)$ (see e.g.~\cite{Kot94}).
The term truncated refers to the suppression of all contours whose weights 
are not damped.
If we replace the original weight $\xi^\xorder$ by the
artificially damped one $\bar{\xi}^\xorder$, we obtain
the following \emph{truncated} partition function,
which can be thought of as an approximation of the
partition function of the $r$-biased random cluster model
with ordered boundary condition:
\begin{align}
	\bar{Z}^{\RC.\order}(\Lambda) &=
			\xe^{-\abs{B(\Lambda)}\cdot e(\xB)}
	        \mathscr{Z}(\Lambda\,|\,\bar\xi^\xorder)\;.
\end{align}
Similarly, replacing $\xi^\xdisorder$ by $\bar{\xi}^\xdisorder$
leads to 
the truncated partition function for
the $r$-biased random cluster model
with disordered boundary condition:
\begin{align}
	\bar{Z}^{\RC.\disorder}(\Lambda) &=
			\xe^{-\abs{B(\Lambda)}\cdot e(\varnothing)}
	        \mathscr{Z}(\Lambda\,|\,\bar\xi^\xdisorder)\;.
\end{align}
The advantage of introducing these truncated partition functions is that
we can apply Proposition~\ref{prop:damped}.
Note that if the original weights are ``damped''
(that is,
$\xi^\xorder(\gamma)\leq \xe^{-\tau|\gamma|}$
or $\xi^\xdisorder(\gamma)\leq \xe^{-\tau|\gamma|}$),
the corresponding
truncated partition functions coincide with the original ones.

From Proposition~\ref{prop:damped} we have the following bounds
for the truncated partition functions:
\begin{align}
	\xe^{\left(g(\bar{\xi}^\xorder)-e(\xB)\right)\abs{B(\Lambda)}
	-C(\tau)\abs{\partial\Lambda}}
	&\leq \makebox[5.5em][c]{$\displaystyle{\bar{Z}^{\RC.\order}(\Lambda)}$}
	\leq \xe^{\left(g(\bar{\xi}^\xorder)-e(\xB)\right)\abs{B(\Lambda)}
	+C(\tau)\abs{\partial\Lambda}}\;, \\
	\xe^{\left(g(\bar{\xi}^\xdisorder)-e(\varnothing)\right)\abs{B(\Lambda)}
	-C(\tau)\abs{\partial\Lambda}}
	&\leq \makebox[5.5em][c]{$\displaystyle{\bar{Z}^{\RC.\disorder}(\Lambda)}$}
	\leq \xe^{\left(g(\bar{\xi}^\xdisorder)-e(\varnothing)\right)\abs{B(\Lambda)}
	+C(\tau)\abs{\partial\Lambda}}\;.
\end{align}

The pressure functions associated to the truncated partition functions are
\begin{align}
	f^\xorder(\beta) &=
		\lim_{n\to\infty}\frac{1}{\abs{B(\Lambda_n)}}
		\log\bar{Z}^{\RC.\order}(\Lambda_n) =
		-e(\xB)+ g(\bar{\xi}^{\xorder}) \;,\\
\label{eq:free-energy:truncated}
	f^\xdisorder(\beta) &=
		\lim_{n\to\infty}\frac{1}{\abs{B(\Lambda_n)}}
		\log\bar{Z}^{\RC.\disorder}(\Lambda_n) =
		-e(\varnothing)+ g(\bar{\xi}^{\xdisorder})\;.
\end{align}

The functions $f^\xorder(\beta)$ and $f^\xdisorder(\beta)$ are lower approximations
of the pressure $f^\RC(\beta)$ of the $r$-biased random-cluster representation.
The next lemma states that in fact when $(q+r)$ is large enough,
the maximum of $f^\xorder$ and $f^\xdisorder$ coincides with $f^\RC$.
As we will see in the next section, for $(q+r)$ large enough,
the functions $g(\bar{\xi}^{\xorder})$ and $g(\bar{\xi}^{\xdisorder})$
and their $\beta$-derivatives are small, and therefore,
the dominant terms of $f^\xorder$ and $f^\xdisorder$ are
$-e(\xB)$ and $-e(\varnothing)$.
This means that $f^\RC$ is approximated by
the maximum of the curves $-e(\xB)$ and $-e(\varnothing)$,
which intersect at a unique value $\beta$, with significantly
different slopes.

By the diameter of a contour $\gamma$, denoted by $\diam\gamma$,
we shall mean the maximum lattice distance between two bonds in $B(\gamma)$.
The next lemma is parallel to Lemma~2 or~\cite{Kot94} or Theorem~3.1 of~\cite{BorImb89}.
\begin{lemma}
\label{lem:damped:naturally}
	Let $(q+r)$ be sufficiently large.
	If $f^\xdisorder\leq f^\xorder$, then
	\begin{enumerate}[ i)]
		\item for every disorder contour $\gamma$ with
			$\diam\gamma\leq\frac{1}{f^\xorder-f^\xdisorder}$ we have
			\begin{align}      
				\xi^\xdisorder(\gamma)\leq \xe^{-\tau|\gamma|}\;,
			\end{align}
		\item for every order contour $\gamma$ we have
			\begin{align}      
				\xi^\xorder(\gamma)\leq \xe^{-\tau|\gamma|}\;.
			\end{align}
	\end{enumerate}
	A similar statement holds if $f^\xorder\leq f^\xdisorder$.
\end{lemma}
\begin{proof}
	We prove the two claims simultaneously by induction on $\diam\gamma$.
	Let $K>0$ and suppose that the claims hold for
	all (disorder/order) contours with diameter less than $K$.
	
	Let $\gamma$ be a disorder contour with diameter $K$
	that satisfies $\diam\gamma\leq \frac{1}{f^\xorder-f^\xdisorder}$.
	Then,
	\begin{align}
		\xi^\xdisorder(\gamma) &=
			\rho(\gamma)\frac{
				Z^{\RC.\order}(\interior\gamma)
			}{
				Z^{\RC.\disorder}(\interior\gamma)
			} \\
		&=
			\rho(\gamma)\frac{
				\bar{Z}^{\RC.\order}(\interior\gamma)
			}{
				\bar{Z}^{\RC.\disorder}(\interior\gamma)
			} \\
		&\leq		
			\rho(\gamma)\frac{
				\xe^{
					f^\xorder\cdot\abs{B(\interior\gamma)}
					+C(\tau)	\vert\partial\interior\gamma\vert
				}
			}{
				\xe^{
					f^\xdisorder\cdot\abs{B(\interior\gamma)}
					-C(\tau)\vert\partial\interior\gamma\vert
				}
			}\\
		&=
			\rho(\gamma)
				\xe^{
					(f^\xorder - f^\xdisorder)\vert B(\interior\gamma)\vert
					+ 2 C(\tau)\vert\partial\interior\gamma\vert
				} \;,
	\end{align}
	where in the second equality we have used the induction hypothesis.
	Namely, every contour in $\interior\gamma$ has diameter less than $K$,
	allowing us to replace the original partition functions
	with the truncated ones.
	Notice that
	\begin{itemize}
		\item $\rho(\gamma)= q (q+r)^{-\frac{1}{4}\abs{\gamma}}$,
		\item $\abs{B(\interior\gamma)}\leq \frac{1}{2}\abs{\gamma}\cdot\diam\gamma$,
		\item $(f^\xorder-f^\xdisorder)\cdot\diam\gamma\leq 1$, and
		\item $\vert\partial\interior\gamma\vert\leq\abs{\gamma}$.
	\end{itemize}
	Hence, we obtain that
	\begin{align}
		\xi^\xdisorder(\gamma) &\leq
			q\,\xe^{-\left(\frac{1}{4}\log(q+r)-1-2C(\tau)\right)\cdot\abs{\gamma}} \;.
	\end{align}
	For $(q+r)$ large enough (uniformly in $\gamma$)
	the righthand side is bounded by $\xe^{-\tau\abs{\gamma}}$,
	hence the claim.

	Next, suppose that $\gamma$ is an order contour
	with diameter $K$.
	We need to show 
	\begin{align}
 		\xi^\xorder(\gamma)=\rho(\gamma) \xe^{\abs{B(\gamma)} e(\xB)}
 		\frac{Z^{\RC.\disorder}(V(\gamma))}{Z^{\RC.\order}(\interior\gamma)}
		\leq \xe^{-\tau|\gamma|}\;.
	\end{align}
	By the induction hypothesis,
	the partition function $Z^{\RC.\order}(\interior\gamma)$
	is equal to the corresponding truncated partition function,
	which can be bounded using Proposition~\ref{prop:damped}.
	As for $Z^{\RC.\disorder}(V(\gamma))$,
	if we suppress all the contours that are ``big'',
	we can get a similar bound using the induction hypothesis.
	
	To render the argument more transparent, we work with the partition functions
	$Y^{\xdisorder}(\Lambda)$ and $Y^{\xdisorder}(\Lambda)$
	(see~(\ref{eq:Z-tilde-f}) and~(\ref{eq:Z-tilde-w})) for which we have
	\begin{align}
		Z^{\RC.\disorder}(\Lambda) &=
			\xe^{-\abs{B(\Lambda)}\cdot e(\varnothing)}
			Y^{\xdisorder}(\Lambda) \;, \\
		Z^{\RC.\order}(\Lambda) &=
			\xe^{-\abs{B(\Lambda)}\cdot e(\xB)}
			Y^{\xorder}(\Lambda) \;.
	\end{align}
	Let us call a disorder contour \emph{small} if
	its diameter is less than or equal to $\frac{1}{f^\xdisorder-f^\xorder}$.
	Otherwise, we call the contour \emph{big}.
	
	As before, let us denote by $\mathcal{E}^\disorder_\Lambda$
	the set of all mutually compatible families of disorder contours in~$\Lambda$ whose elements are external.
	Factoring the contribution of the interior of big external contours we have the recursion\footnote{%
		Although $\Lambda\setminus\interior\theta$ does not match our requirement for
		being a volume (i.e., not having holes), it does not cause any problem.
		In fact, since the contours in $Y^\xdisorder_\SMALL$ are small,
		they cannot surround the holes in $\Lambda\setminus\interior\theta$,
		hence they do not distinguish the holes from the outside of $\Lambda$.
	}
	\begin{align}
		Y^\xdisorder(\Lambda) &=
			\sum_{
				\substack{\theta\in\mathcal{E}^\disorder_\Lambda \\
					\text{$\theta$ big}
				}
			}
			Y^\xdisorder_\SMALL(
				\Lambda\setminus\interior\theta
			)
			\prod_{\gamma'\in\theta}
			\widetilde{\rho}(\gamma')
			Y^\xorder(\interior\gamma')
	\end{align}
	where $\interior\theta\IsDef\bigcup_{\gamma'\in\theta}\interior\gamma'$,
	and $Y^\xdisorder_\SMALL(\Lambda)\IsDef\mathscr{Z}(\Lambda\,|\,\xi^\xdisorder_\SMALL)$,
	in which the weight function $\xi^\xdisorder_\SMALL$ is obtained from $\xi^\xdisorder$ by replacing
	the weights of all big contours with $0$, that is,
	\begin{align}
		\xi^\xdisorder_\SMALL(\gamma') &= \begin{cases}
			\xi^\xdisorder(\gamma') \qquad & \text{if $\gamma'$ small,} \\
			0 & \text{if $\gamma'$ big.}
		\end{cases}
	\end{align}
	
	The expression for the weight $\xi^\xorder(\gamma)$ reads then
	\begin{align}
		\xi^\xorder(\gamma) &=
			\rho(\gamma)\cdot \xe^{\abs{B(V(\gamma))}\cdot\left(e(\xB)-e(\varnothing)\right)}
			\sum_{
				\substack{\theta\in\mathcal{E}^\disorder_{V(\gamma)} \\
					\text{$\theta$ big}
				}
			}
			\frac{%
				Y^\xdisorder_\SMALL(V(\gamma)\setminus\interior\theta) \cdot
				Y^\xorder(\interior\theta)
			}{%
				Y^\xorder(\interior\gamma)
			}
			\prod_{\gamma'\in\theta}
			\widetilde{\rho}(\gamma') \;.
	\end{align}
	The induction hypothesis and Proposition~\ref{prop:damped} tell us:
	\begin{itemize}
		\item $Y^\xdisorder_\SMALL(V(\gamma)\setminus\interior\theta) \leq
			\xe^{g(\xi^\xdisorder_{\SMALL})\cdot\abs{B(V(\gamma)\setminus\interior\theta)}+
			C(\tau)\cdot\abs{\partial \left(V(\gamma)\setminus\interior\theta\right)}}$,
		\item $Y^\xorder(\interior\theta) \leq
			\xe^{g(\bar{\xi}^\xorder)\cdot\abs{B(\interior\theta)}+ C(\tau)\cdot\abs{\partial\interior\theta}}$,
		\item $Y^\xorder(\interior\gamma) \geq
			\xe^{g(\bar{\xi}^\xorder)\cdot\abs{B(\interior\gamma)}- C(\tau)\cdot\abs{\partial\interior\gamma}}$, and
		\item $g(\bar{\xi}^\xorder)\leq \xe^{-\tau/2} \leq 1$.
	\end{itemize}
	Moreover
	\begin{itemize}
		\item $\abs{\partial \left(V(\gamma)\setminus\interior\theta\right)}\leq
			\abs{\partial V(\gamma)} + \abs{\partial\interior\theta} \leq
			3\abs{\gamma} + \sum_{\gamma'\in\theta}\abs{\gamma'}$, and
		\item $\abs{B(\gamma)}\leq\abs{\gamma}$.
	\end{itemize}
	Hence we have
	\begin{align}
		\xi^\xorder(\gamma) &\leq
			\rho(\gamma)\cdot\xe^{\left(1+4 C(\tau)\right)\cdot\abs{\gamma}}
			\sum_{
				\substack{\theta\in\mathcal{E}^\disorder_{V(\gamma)} \\
					\text{$\theta$ big}
				}
			}
			\xe^{\abs{B(V(\gamma)\setminus\interior\theta)}\cdot\left[
				e(\xB)-e(\varnothing)+g(\xi^\xdisorder_\SMALL)-g(\bar{\xi}^\xorder)
			\right]}
			\prod_{\gamma'\in\theta}\rho(\gamma')\cdot \xe^{2 C(\tau)\cdot\abs{\gamma'}} \\
		&=
			\rho(\gamma)\cdot\xe^{\left(1+4 C(\tau)\right)\cdot\abs{\gamma}}
			\sum_{
				\substack{\theta\in\mathcal{E}^\disorder_{V(\gamma)} \\
					\text{$\theta$ big}
				}
			}
			\xe^{\abs{B(V(\gamma)\setminus\interior\theta)}\cdot\left[
				\left(g(\xi^\xdisorder_\SMALL)-g(\bar{\xi}^\xdisorder)\right) -
				\left(f^\xorder - f^\xdisorder\right)
			\right]}
			\prod_{\gamma'\in\theta}\rho(\gamma')\cdot \xe^{2 C(\tau)\cdot\abs{\gamma'}} \;.
	\end{align}
	As we shall see shortly,
	if $(q+r)$ is large enough, the sum appearing in the above expression can be bounded by
	$\xe^{3C(\tau)\cdot\abs{\gamma}}$, so that
	\begin{align}
		\xi^\xorder(\gamma) &\leq
			\rho(\gamma)\cdot\xe^{\left(1+4 C(\tau)\right)\cdot\abs{\gamma}}\xe^{3C(\tau)\cdot\abs{\gamma}} \\
		&=
			\xe^{-\left(\frac{1}{4}\log(q+r)-1-7C(\tau)\right)\cdot\abs{\gamma}}
	\end{align}
	For large $(q+r)$, the righthand side is bounded by $\xe^{-\tau\abs{\gamma}}$, proving the claim.
	
	It remains to show that for $(q+r)$ sufficiently large,
	\begin{align}
		\sum_{
				\substack{\theta\in\mathcal{E}^\disorder_{V(\gamma)} \\
					\text{$\theta$ big}
				}
			}
			\xe^{\abs{B(V(\gamma)\setminus\interior\theta)}\cdot\left[
				\left(g(\xi^\xdisorder_\SMALL)-g(\bar{\xi}^\xdisorder)\right) -
				\left(f^\xorder - f^\xdisorder\right)
			\right]}
			\prod_{\gamma'\in\theta}\rho(\gamma')\cdot \xe^{2 C(\tau)\cdot\abs{\gamma'}}
		&\leq
			\xe^{3C(\tau)\cdot\abs{\gamma}} \;.
	\end{align}
	To show this, let us consider a contour model with weight function
	\begin{align}
		\widehat{\rho}(\gamma') &= \begin{cases}
			\rho(\gamma')\cdot\xe^{3C(\tau)\cdot\abs{\gamma'}}\;, \qquad & \text{if $\gamma'$ big and disorder,} \\
			0\;, & \text{otherwise.}
		\end{cases}
	\end{align}
	Assuming that
	$g(\widehat{\rho})\leq f^\xorder - f^\xdisorder
		+g(\bar{\xi}^\xdisorder)-g(\xi^\xdisorder_\SMALL)$,
	and $(q+r)$ in such a way that
	$\widehat{\rho}(\gamma')\leq\xe^{-\tau\abs{\gamma'}}$,
	we can use Proposition~\ref{prop:damped} to obtain
	\begin{align}
		&
		\sum_{
				\substack{\theta\in\mathcal{E}^\disorder_{V(\gamma)} \\
					\text{$\theta$ big}
				}
			}
			\xe^{\abs{B(V(\gamma)\setminus\interior\theta)}\cdot\left[
				\left(g(\xi^\xdisorder_\SMALL)-g(\bar{\xi}^\xdisorder)\right) -
				\left(f^\xorder - f^\xdisorder\right)
			\right]}
			\prod_{\gamma'\in\theta}\rho(\gamma')\cdot \xe^{2 C(\tau)\cdot\abs{\gamma'}} \nonumber \\
		&\leq
			\sum_{
				\substack{\theta\in\mathcal{E}^\disorder_{V(\gamma)} \\
					\text{$\theta$ big}
				}
			}
			\xe^{-\abs{B(V(\gamma)\setminus\interior\theta)}\cdot g(\widehat{\rho})}
			\prod_{\gamma'\in\theta}\rho(\gamma')\cdot \xe^{2 C(\tau)\cdot\abs{\gamma'}} \\
		&=
			\xe^{-\abs{B(V(\gamma))}\cdot g(\widehat{\rho})}
			\sum_{
				\substack{\theta\in\mathcal{E}^\disorder_{V(\gamma)} \\
					\text{$\theta$ big}
				}
			}
			\prod_{\gamma'\in\theta}\rho(\gamma')\cdot \xe^{2 C(\tau)\cdot\abs{\gamma'}}
				\cdot \xe^{\abs{B(\interior\gamma')}\cdot g(\widehat{\rho})} \\
		&\leq
			\xe^{-\abs{B(V(\gamma))}\cdot g(\widehat{\rho})}
			\sum_{
				\substack{\theta\in\mathcal{E}^\disorder_{V(\gamma)} \\
					\text{$\theta$ big}
				}
			}
			\prod_{\gamma'\in\theta}\rho(\gamma')\cdot \xe^{2 C(\tau)\cdot\abs{\gamma'}}
				\cdot \mathscr{Z}(\interior\gamma'\,|\,\widehat{\rho})
				\cdot \xe^{C(\tau)\cdot\abs{\gamma'}} \\
		&=
			\xe^{-\abs{B(V(\gamma))}\cdot g(\widehat{\rho})}
			\cdot \mathscr{Z}(V(\gamma)\,|\,\widehat{\rho}) \\
		&\leq
			\xe^{3C(\tau)\cdot\abs{\gamma}} \;.
	\end{align}
	
	Finally, to see that $g(\widehat{\rho})\leq f^\xorder - f^\xdisorder
		+g(\bar{\xi}^\xdisorder)-g(\xi^\xdisorder_\SMALL)$,
	note that $g(\bar{\xi}^\xdisorder)-g(\xi^\xdisorder_\SMALL)\geq 0$ (due to the fact that
	$\xi^\xdisorder_\SMALL\leq\bar{\xi}^\xdisorder$) and by Proposition~\ref{prop:damped}
	\begin{align}
		g(\widehat{\rho}) &\leq
			\sum_{\substack{\gamma':S(\gamma')\ni 0 \\ \text{$\gamma'$ big}}} \widehat{\rho}(\gamma')
		\leq
			\xe^{-\tau/(f^\xorder-f^\xdisorder)} \;,
	\end{align}
	using the fact that $\gamma'$ is big only if $\abs{\gamma'}\geq \frac{2}{f^\xorder-f^\xdisorder}$.
	For $\tau$ not too small we have $\xe^{-\tau/(f^\xorder-f^\xdisorder)}\leq f^\xorder-f^\xdisorder$.
\end{proof}

\begin{figure}[h]
	\begin{center}
		\begin{tabular}{cc}
		\begin{minipage}[b]{0.45\textwidth}
		\begin{center}
		\begin{tikzpicture}[line cap=round,>=stealth',xscale=1.3,yscale=1.1]
			
			\def\mthck{0.35}
			\def\gdiffo{0.18}
			\def\gdiffd{0.27}
			\def\bcrit{2.05}
			\def\slope{1}
			\def\ebase{1.3}
			
			\def\xmin{1.12}
			\def\xmax{1.96}
			\def\ymin{-0.45}
			\def\ymax{0.3}
			
			\draw[dotted] (\xmin,\ymin) rectangle (\xmax,\ymax);

			\draw[->,black!80] (-0.1,0) -- (4.1,0) node[below] {$\beta$};
			\draw[->,black!80] (0,-2.1) -- (0,2.1);
						
			\fill[color=red!30,samples=50,opacity=0.5]
				plot[domain=0.31:3.5] (\x,{ln(1-exp(-\slope*\x))}) --
				(3.5,{ln(1-exp(-\slope*3.5))+\mthck}) --
				plot[domain=3.5:0.31] (\x,{ln(1-exp(-\slope*\x))+\mthck}) --
				cycle;
			
			\draw[<->,ultra thin]
				(3.65,{ln(1-exp(-\slope*3.5))}) --
					node[inner sep=0pt,above] {\hspace{2.5em} \scriptsize $\xe^{-\tau/2}$}
				(3.65,{ln(1-exp(-\slope*3.5))+\mthck});
			
			\fill[color=blue!30,opacity=0.5]
				plot[domain=3.2:0] (\x,{\ebase-\slope*\x}) --
				(0,{\ebase+\mthck}) --
				plot[domain=0:3.2] (\x,{\ebase+\mthck-\slope*\x}) --
				cycle;
			
			\draw[<->,ultra thin]
				(-0.15,\ebase) -- node[left] {\scriptsize $\xe^{-\tau/2}$}
				(-0.15,\ebase+\mthck);

			\draw[color=red,thick,samples=50]
				plot[domain=3.5:0.31] (\x,{ln(1-exp(-\slope*\x))})
				node[below] {\scriptsize $-e(\xB)$};

			\draw[color=blue,thick] 
				plot[domain=0:3.2] (\x,{\ebase-\slope*\x})
				node[below] {\scriptsize $-e(\varnothing)$};
						
			\draw[color=black,samples=50]
				plot[domain=3.5:0.31] (\x,{ln(1-exp(-\slope*\x))+\gdiffo+0.005*cos(1000*\x)})
				node[inner sep=0pt,left] {\scriptsize $f^\xorder$};

			\draw[color=black,samples=50]
				plot[domain=0:3.2] (\x,{\ebase-\slope*\x+\gdiffd-0.1+0.01*sin(600*\x)})
				node[inner sep=1pt,right] {\scriptsize $f^\xdisorder$};
						
		\end{tikzpicture}
		\end{center}
		
		
		\end{minipage}
		
		&
		
		\begin{minipage}[b]{0.45\textwidth}			
		\begin{center}
		\begin{tikzpicture}[line cap=round,>=stealth',xscale={1.3*4.5},yscale={1.1*4.5}]
			
			\def\mthck{0.35}
			\def\gdiffo{0.18}
			\def\gdiffd{0.27}
			\def\bcrit{2.05}
			\def\slope{1}
			\def\ebase{1.3}
			\def\dslope{0.15}
			
			\def\xmin{1.12}
			\def\xmax{1.96}
			\def\ymin{-0.45}
			\def\ymax{0.3}
			\def\ddd{0.005}
			\def\xcent{1.53}
			\def\ycent{-0.064}
			
			\begin{scope}
				\clip (\xmin-\ddd,\ymin-\ddd) rectangle (\xmax+\ddd,\ymax+\ddd);
				\draw[dotted] (\xmin,\ymin) rectangle (\xmax,\ymax);
							
%
				
				\begin{scope}
					\clip
						plot[domain=3.2:0] (\x,{\ebase-\slope*\x}) --
						(0,{\ebase+\mthck}) --
						plot[domain=0:3.2] (\x,{\ebase+\mthck-\slope*\x}) --
						cycle;
					
					\clip
						plot[domain=0.31:3.5] (\x,{ln(1-exp(-\slope*\x))}) --
						(3.5,{ln(1-exp(-\slope*3.5))+\mthck}) --
						plot[domain=3.5:0.31] (\x,{ln(1-exp(-\slope*\x))+\mthck}) --
						cycle;
					
					\fill[color=red!30,opacity=0.5]
						(\xmin+\ddd,\ymin+\ddd) rectangle (\xmax-\ddd,\ymax-\ddd);
					\fill[color=blue!30,opacity=0.5]
						(\xmin+\ddd,\ymin+\ddd) rectangle (\xmax-\ddd,\ymax-\ddd);
				\end{scope}
				
%
%
				
				\draw[color=black,samples=50]
					plot[domain=3.5:0.31] (\x,{ln(1-exp(-\slope*\x))+\gdiffo+0.005*cos(1000*\x)})
					node[inner sep=0pt,left] {\scriptsize $f^\xorder$};
	
				\draw[color=black,samples=50]
					plot[domain=0:3.2] (\x,{\ebase-\slope*\x+\gdiffd-0.1+0.01*sin(600*\x)})
					node[inner sep=1pt,right] {\scriptsize $f^\xdisorder$};

				\begin{scope}
					\clip (\xcent,\ycent) circle (10pt);
									
					\def\oslope{\slope*exp(-\slope*\xcent)/(1-exp(-\slope*\xcent))}
					\fill[color=gray,opacity=0.7]
						plot[domain=\xmin:\xmax] (\x,{(\oslope+\dslope)*(\x-\xcent)+\ycent}) --
						plot[domain=\xmax:\xmin] (\x,{(\oslope-\dslope)*(\x-\xcent)+\ycent}) --
						cycle;
					
					\fill[color=gray,opacity=0.7]
						plot[domain=\xmin:\xmax] (\x,{(-\slope+\dslope)*(\x-\xcent)+\ycent}) --
						plot[domain=\xmax:\xmin] (\x,{(-\slope-\dslope)*(\x-\xcent)+\ycent}) --
						cycle;
					
				\end{scope}
			\end{scope}
			
			\draw (\xmin,{ln(1-exp(-\slope*\xmin))+\gdiffo+0.005*cos(1000*\xmin)})
				node[inner sep=1pt,left] {\scriptsize $f^\xorder$};
			\draw (\xmax,\ymin)
				node[inner sep=1pt,below] {\scriptsize $f^\xdisorder$};

		\end{tikzpicture}
		\end{center}
		
		\vspace{0.2ex}
		
		\end{minipage}
		
		\\
		
		(a) & (b)
		\end{tabular}
	\end{center}
	\caption{
		(a) The curves of $f^\xorder(\beta)$ and $f^\xdisorder(\beta)$
		are within a narrow margin above $-e(\xB)$ and $-e(\varnothing)$.
		(b) The slopes of $f^\xorder(\beta)$ and $f^\xdisorder(\beta)$
		are close to the slopes of $-e(\xB)$ and $-e(\varnothing)$.
	}
	\label{fig:intersection}
\end{figure}
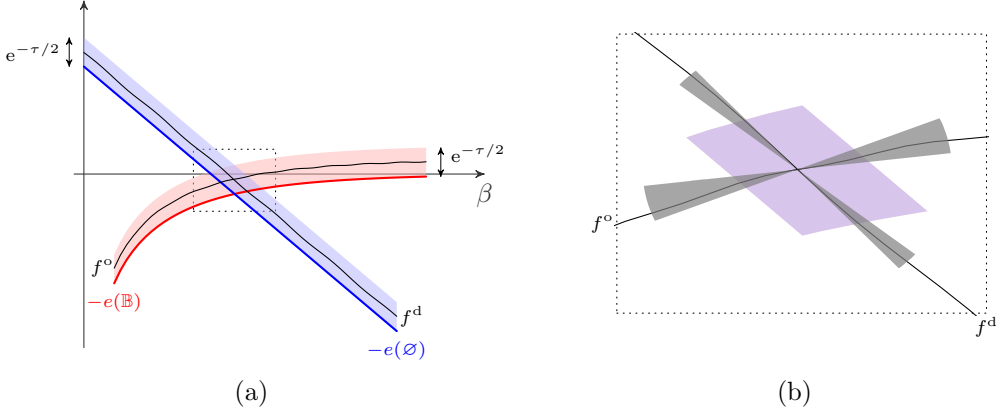

From the above lemma, we know that the pressure $f^\RC(\beta)$ of
the $r$-biased random-cluster representation
coincides with maximum of the functions $f^\xorder(\beta)$ and $f^\xdisorder(\beta)$, that is,
\begin{align}
	f^\RC(\beta) &= \max\left\{f^\xorder(\beta),f^\xdisorder(\beta)\right\} \;.
\end{align}
Recall that
\begin{align}
	f^\xorder(\beta) &= -e(\xB) + g(\bar{\xi}^\xorder) \;, \\
	f^\xdisorder(\beta) &= -e(\varnothing) + g(\bar{\xi}^\xdisorder) \;.
\end{align}
If $\tau$ is large, Proposition~\ref{prop:damped} says that $g(\bar{\xi}^\xorder)$
and $g(\bar{\xi}^\xdisorder)$ are small, so that $f^\RC(\beta)$ can be nearly expressed
in terms of
the ``energy'' per bond of the fully ordered and fully disordered configurations.
More precisely, if we define
\begin{align}
	F(\beta) &\IsDef \max\left\{-e(\xB),-e(\varnothing)\right\} \;.
\end{align}
we have
\begin{align}
\label{eq:tube}
	0 \leq f^\RC(\beta) - F(\beta) \leq \xe^{-\tau/2} \;.
\end{align}

The two curves $-e(\xB)$ and $-e(\varnothing)$ (as functions of $\beta$)
intersect at a single point
\begin{align}
	\bar{\beta}_\critical &=
		\log\left(1+\sqrt{q+r}\right) \;,
\end{align}
above which $F(\beta)=-e(\xB)$ and below which $F(\beta)=-e(\varnothing)$.
Furthermore, these two curves have significantly different slopes, implying that
$F(\beta)$ is not differentiable at $\bar{\beta}_\critical$.
What we are after is to infer that $f^\RC(\beta)$ has a similar behavior.  In other words, we would like to show that
there exists a unique solution $\beta_\critical$ for the equation $f^\xorder(\beta)=f^\xdisorder(\beta)$,
at which $f^\RC(\beta)$ is not differentiable,
above which $f^\RC(\beta)=f^\xorder(\beta)$ and below which $f^\RC(\beta)=f^\xdisorder(\beta)$.
Note that condition~(\ref{eq:tube}) guarantees that
$f^\RC(\beta)=f^\xorder(\beta)>f^\xdisorder(\beta)$ for $\beta\gg \bar{\beta}_\critical$
and $f^\RC(\beta)=f^\xdisorder(\beta)>f^\xorder(\beta)$ for $\beta\ll\bar{\beta}_\critical$.
In fact, it states that $f^\RC(\beta)$ lives in a margin of width $\xe^{-\tau/2}$
above $F(\beta)$ (see Figure~\ref{fig:intersection}(a)).
To infer such a sharp transition we further need to give bounds for the derivatives of
$g(\bar{\xi}^\xorder)$ and $g(\bar{\xi}^\xdisorder)$ (see Figure~\ref{fig:intersection}(b)).
This is addressed in the following lemma, which is analogous to Theorem~3.3 of~\cite{BorImb89}.

\begin{lemma}
\label{lem:derivative:bound}
	We have
	\begin{align}
		\frac{\partial}{\partial\beta}g(\bar{\xi}^\xorder)
		&\leq
			-2\,\frac{\partial e(\xB)}{\partial\beta}
			\sum_{\gamma: S(V(\gamma))\ni 0} \bar{\xi}^\xorder(\gamma) \leq
			\frac{2}{\xe^\beta-1}\, \xe^{-\tau/2} \;, \\
		\frac{\partial}{\partial\beta}g(\bar{\xi}^\xdisorder)
		&\leq
			-2\,\frac{\partial e(\xB)}{\partial\beta}
			\sum_{\gamma: S(\interior\gamma)\ni 0} \bar{\xi}^\xdisorder(\gamma) \leq
			\frac{2}{\xe^\beta-1}\, \xe^{-\tau/2} \;.
	\end{align}
\end{lemma}
\begin{proof}
	For a finite volume $\Lambda$,
	if we denote by $\Gamma(\Lambda)$ the set of all contours in $\Lambda$,
	we have
	\begin{align}
		\frac{\partial}{\partial\beta}\log\mathscr{Z}(\Lambda\,|\,\bar{\xi}^\xorder) &=
			\frac{1}{\mathscr{Z}(\Lambda\,|\,\bar{\xi}^\xorder)}
			\sum_{\partial\in\mathcal{M}_\Lambda}
			\frac{\partial}{\partial\beta}
			\prod_{\gamma\in\partial}\bar{\xi}^\xorder(\gamma) \\
		&=
			\frac{1}{\mathscr{Z}(\Lambda\,|\,\bar{\xi}^\xorder)}
			\sum_{\partial\in\mathcal{M}_\Lambda}
			\sum_{\gamma\in\partial}
			\frac{\partial\bar{\xi}^\xorder(\gamma)}{\partial\beta}
			\prod_{\substack{\widehat{\gamma}\in\partial\\ \widehat{\gamma}\neq\gamma}}\bar{\xi}^\xorder(\widehat{\gamma}) \\
		&=
			\sum_{\gamma\in\Gamma(\Lambda)}
			\frac{\partial\bar{\xi}^\xorder(\gamma)}{\partial\beta}
			\frac{
				\mathscr{Z}(\interior\gamma\,|\,\bar{\xi}^\xorder)\cdot
				\mathscr{Z}(\Lambda\cap\exterior\gamma\,|\,\bar{\xi}^\xorder)
			}{
				\mathscr{Z}(\Lambda\,|\,\bar{\xi}^\xorder)
			} \\
		&\leq
			\sum_{\gamma\in\Gamma(\Lambda)}
			\frac{\partial\bar{\xi}^\xorder(\gamma)}{\partial\beta} \;.
	\end{align}
	Recall that $\bar{\xi}^\xorder(\gamma)$ is equal to either $0$ or $\xi^\xorder(\gamma)$,
	hence the derivative of $\bar{\xi}^\xorder(\gamma)$ is bounded by
	the derivative of $\xi^\xorder(\gamma)$.
	Using the definition of $\xi^\xorder(\gamma)$ (Eq.~(\ref{eq:xi-order})) we have
	\begin{align}
		\frac{\partial\xi^\xorder(\gamma)}{\partial\beta} &=
			\abs{B(\gamma)}\cdot\frac{\partial e(\xB)}{\partial\beta}\cdot\xi^\xorder(\gamma) \ + \ 
			\rho(\gamma)\cdot\xe^{\abs{B(\gamma)}\cdot e(\xB)}\cdot
			\frac{\partial}{\partial\beta}
			\frac{%
				Z^{\RC.\disorder}(V(\gamma))
			}{%
				Z^{\RC.\order}(\interior\gamma)
			} \;.
	\end{align}
	The derivative of the partition functions appearing on the righthand side
	can be bounded directly from the definitions (Eq.~(\ref{eq:rc-disordered-def}) and~(\ref{eq:rc-ordered-def})) by
	\begin{align}
		0 \leq \frac{\partial}{\partial\beta}Z^{\RC.\disorder}(V(\gamma)) &\leq
			-\abs{B(V(\gamma))}\cdot\frac{\partial e(\xB)}{\partial\beta}\cdot
			Z^{\RC.\disorder}(V(\gamma)) \;, \\
		0 \leq \frac{\partial}{\partial\beta}Z^{\RC.\order}(\interior\gamma) &\leq
			-\abs{B(\interior\gamma)}\cdot\frac{\partial e(\xB)}{\partial\beta}\cdot
			Z^{\RC.\order}(\interior\gamma) \;,
	\end{align}
	leading to
	\begin{align}
		\frac{\partial}{\partial\beta}
			\frac{%
				Z^{\RC.\disorder}(V(\gamma))
			}{%
				Z^{\RC.\order}(\interior\gamma)
			} &\leq
			-\abs{B(V(\gamma))}\cdot\frac{\partial e(\xB)}{\partial\beta}\cdot
			\frac{%
				Z^{\RC.\disorder}(V(\gamma))
			}{%
				Z^{\RC.\order}(\interior\gamma)
			} \;.
	\end{align}
	Therefore,
	\begin{align}
		\frac{\partial\bar{\xi}^\xorder(\gamma)}{\partial\beta} &\leq
			-\abs{B(V(\gamma))}\cdot\frac{\partial e(\xB)}{\partial\beta}\cdot
			\bar{\xi}^\xorder(\gamma) \;.
	\end{align}
	
	We can now write 
	\begin{align}
		\frac{\partial}{\partial\beta}\frac{1}{\abs{S(\Lambda)}}\log\mathscr{Z}(\Lambda\,|\,\bar{\xi}^\xorder) &\leq
			-\frac{\partial e(\xB)}{\partial\beta}\cdot \frac{1}{\abs{S(\Lambda)}}
			\sum_{\gamma\in\Gamma(\Lambda)}
			\abs{B(V(\gamma))}\cdot \bar{\xi}^\xorder(\gamma) \\
		&=
			-\frac{\partial e(\xB)}{\partial\beta}\cdot \frac{1}{\abs{S(\Lambda)}}
			\sum_{\gamma\in\Gamma(\Lambda)} \sum_{b\in B(V(\gamma))} \bar{\xi}^\xorder(\gamma) \\
		&=
			-\frac{\partial e(\xB)}{\partial\beta}\cdot \frac{1}{\abs{S(\Lambda)}}
			\sum_{b\in B(\Lambda)} \sum_{\substack{\gamma\in\Gamma(\Lambda)\\ B(V(\gamma))\ni b}}
			\bar{\xi}^\xorder(\gamma) \\
		&\leq
			-\frac{\partial e(\xB)}{\partial\beta}\cdot \frac{1}{\abs{S(\Lambda)}}
			\sum_{b\in B(\Lambda)} \sum_{\substack{\gamma\in\Gamma\\ B(V(\gamma))\ni b}}
			\bar{\xi}^\xorder(\gamma) \\
		&\leq
			-\frac{\partial e(\xB)}{\partial\beta}\cdot
			\frac{\abs{B(\Lambda)}}{\abs{S(\Lambda)}}
			\sum_{\substack{\gamma\in\Gamma\\ S(V(\gamma))\ni 0}}
			\bar{\xi}^\xorder(\gamma) \\
		&\leq
			-2\,\frac{\partial e(\xB)}{\partial\beta}
			\sum_{\substack{\gamma\in\Gamma\\ S(V(\gamma))\ni 0}}
			\bar{\xi}^\xorder(\gamma) \;.
	\end{align}
	Since $\tau$ is large, the above sum converges and leads to a uniform bound (with respect to $\Lambda$)
	for $\frac{\partial}{\partial\beta}\frac{1}{\abs{S(\Lambda)}}\log\mathscr{Z}(\Lambda\,|\,\bar{\xi}^\xorder)$.
	Using the dominated convergence theorem the first claim follows.
	The proof of the other claim is similar.
\end{proof}

\section{The Main Result}
\label{sec:main}


\begin{theorem}
	Let $\varepsilon>0$ and $(q+r)$ large enough.
	The two-dimensional $(q,r)$-Potts model
	undergoes a first-order transition in temperature with breaking of permutation symmetry:
	\begin{enumerate}[ i)]
		\item Above the transition temperature, the model has a unique Gibbs state $\mu^\free$,
			which is ``disordered''.
		\item Below the transition temperature, there exist at least $q$
			different ``ordered'' Gibbs states $\mu^1,\mu^2,\ldots,\mu^q$.
		\item At the transition temperature, $q$ ``ordered'' Gibbs states $\mu^1,\mu^2,\ldots,\mu^q$
			coexist with a ``disordered'' Gibbs state $\mu^\free$.
	\end{enumerate}
	The ``ordered'' and ``disordered'' states can be distinguished by
	\begin{align}
		&\mu^k(\left\{\sigma: \sigma_i=k\right\}) > 1-\varepsilon\,,&&\text{for every visible $k$,} \\
		&\mu^\free(\left\{\sigma: \sigma_i=k\right\})< \varepsilon\,,&&\text{for every $k$,}
	\end{align}
	for every site $i$ in the lattice.
\end{theorem}

\begin{proof}
%
	In the previous sections we have introduced the main ingredients to prove
	the occurrence of a first-order transition.
	Below, we first put these ingredients together so as to obtain a recipe
	for the proof.
	Afterwards, we shall see how these ingredients along with basic properties
	of the biased random-cluster model (see Appendix~\ref{apx:rc:properties})
	can be used to prove the symmetry breaking at the transition temperature.
	
	The first step was to reduce the partition function of the $(q,r)$-Potts model
	to the partition function of the $r$-biased random-cluster model.
	This was done for the free and homogeneous visible boundary conditions,
	which led to the disordered and ordered boundary conditions for the $r$-biased random-cluster model
	(see Eq.~(\ref{eq:potts-rc:disordered}) and~(\ref{eq:potts-rc:ordered})).
	By means of this we could rewrite the pressure $f(\beta)$ for the $(q,r)$-Potts model as
	\begin{align}
		f(\beta) &=
			2\beta+2f^\RC(\beta) \;,
	\end{align}
	where
	\begin{align}
		f^\RC(\beta) &=
			\lim_{n\to\infty} \frac{\log Z^{\RC.\order}_{p_\beta,q,r}(\Lambda_n)}{\abs{B(\Lambda_n)}} =
			\lim_{n\to\infty} \frac{\log Z^{\RC.\disorder}_{p_\beta,q,r}(\Lambda_n)}{\abs{B(\Lambda_n)}} \;.
	\end{align}

	The second step consisted in re-expressing the partition functions
	$Z^{\RC.\order}$ and $Z^{\RC.\disorder}$ in terms of the partition functions of two
	abstract contour models $\mathscr{Z}(\cdot\,|\,\xi^\xorder)$ and
	$\mathscr{Z}(\cdot\,|\,\xi^\xdisorder)$ (Lemma~\ref{lem:contour-model})
	so that we could write
	\begin{align}
		f^\RC(\beta) &=
			-e(\xB) + g(\xi^\xorder) =
			-e(\varnothing) + g(\xi^\xdisorder) \;,
	\end{align}
	where $g(\xi^\xorder)$ and $g(\xi^\xdisorder)$ are the pressure functions
	for the two contour models (Eq.~(\ref{eq:pressure:contour-model})).
	
	If the weight function $\chi$ of a contour model is sufficiently ``damped''
	(i.e., $\chi(\gamma)\leq\xe^{-\tau\abs{\gamma}}$ for $\tau$ large enough),
	the corresponding pressure $g(\chi)$ can be made arbitrarily small (Proposition~\ref{prop:damped}).
	In order to exploit this result, in the third step, we truncated the weight functions
	$\xi^\xorder$ and $\xi^\xdisorder$ so as to render them
	artificially damped (Eq.~(\ref{eq:weight:truncated})).
	We could then define two functions
	\begin{align}
		f^\xorder(\beta) &=	-e(\xB) + g(\bar{\xi}^\xorder) \;, \\
		f^\xdisorder(\beta) &=	-e(\varnothing) + g(\bar{\xi}^\xdisorder) \;,
	\end{align}
	which approximate $f^\RC(\beta)$ from below, and which can be thought of
	(for sufficiently large $\tau$) as
	perturbations of the functions $-e(\xB)$
	and $-e(\varnothing)$, respectively (see Figure~\ref{fig:intersection}(a)).
	
	In the fourth step, we proved that, for $(q+r)$ large (relative to $\tau$),
	the pressure $f^\RC(\beta)$ is the maximum of
	these two approximations.
	This was achieved by proving that whenever $f^\xorder\geq f^\xdisorder$,
	the weight function $\xi^\xorder$ is ``naturally'' damped (i.e., $\xi^\xorder=\bar{\xi}^\xorder$)
	and vice versa (Lemma~\ref{lem:damped:naturally}).
	Therefore, $f^\RC(\beta)$ can be closely approximated by the maximum between
	$-e(\xB)$ and $-e(\varnothing)$.
	Due to the continuity of the pressure functions, the latter implies that the curves
	$f^\xorder(\beta)$ and $f^\xdisorder(\beta)$ intersect.
	
	In the last step, we showed that for $\tau$ sufficiently large,
	$\frac{\partial}{\partial\beta}f^\xorder(\beta) > \frac{\partial}{\partial\beta}f^\xdisorder(\beta)$
	(Lemma~\ref{lem:derivative:bound}).
	Therefore, the two functions meet at a unique point $\beta_\critical$
	at which $f^\RC(\beta)$ is non-differentiable (Figure~\ref{fig:intersection}(b)).
	Hence, the $(q,r)$-Potts model undergoes a first-order phase transition
	at $\beta_\critical$.
	
	We now prove the breaking of permutation symmetry at the transition temperature.
	Note that if $\beta\geq\beta_\critical$, we have $f^\xorder(\beta)\geq f^\xdisorder(\beta)$
	and therefore, in view of Lemma~\ref{lem:damped:naturally}, the weights $\xi^\xorder$
	satisfy $\xi^\xorder(\gamma)\leq\xe^{-\tau\abs{\gamma}}$.
	Since $\tau$ was chosen large, for every finite set of sites $A$, the sum
	$\sum_{\gamma\in\Gamma_A} \xi^\xorder(\gamma)$
	converges.  Hence, it follows from Corollary~\ref{cor:peierls}
	that
	\begin{align}
	\label{eq:thm:sea-of-order}
		\phi^\order_{p_\beta}\left(
			\begin{array}{c}
				\text{\rm $\exists$ unique infinite sea of order} \\
				\text{\rm with finite islands of disorder}
			\end{array}
		\right)=1 \;.
	\end{align}
	(The subscript $p_\beta$ is added to emphasize the dependence on $\beta$.)
	Similarly, if $\beta\leq\beta_\critical$, the sum $\sum_{\gamma\in\Gamma_A} \xi^\xdisorder(\gamma)$
	converges and thus
	\begin{align}
	\label{eq:thm:sea-of-disorder}
		\phi^\disorder_{p_\beta}\left(
			\begin{array}{c}
				\text{\rm $\exists$ unique infinite sea of disorder} \\
				\text{\rm with finite islands of order}
			\end{array}
		\right)=1 \;.
	\end{align}
		
	For a visible colour $k$,
	let $\mu^k_\beta$ be, as introduced in Section~\ref{sec:model},
	a weak limit of Boltzmann distributions with
	homogeneous boundary conditions $\omega^k$.
	Likewise, let $\mu^\free_\beta$ be a weak limit
	of free-boundary Boltzmann distributions.
	
	Similar to the finite-volume couplings, there exists
	a coupling of $\mu^k_\beta$ and $\phi^\order_{p_\beta}$
	with the property that
	with probability~$1$,
	every site incident to an infinite connected component of present bonds
	is coloured with $k$ (see Appendix~\ref{apx:rc:properties}).
	If $\beta\geq\beta_\critical$, we know that almost surely
	the bond configuration consists of a unique sea of order with finite
	islands of disorder.
	In particular,
	the probability that a given site $i$ takes a colour other than $k$
	is bounded by the probability that site $i$
	is surrounded by an order contour; that is,
	\begin{align}
		\mu^k_\beta(\left\{\sigma:\sigma_i\neq k\right\}) &\leq
			\phi^\order_{p_\beta}\{
				X: \xpartial X\cap\Gamma_i\neq\varnothing
			\} \;.
	\end{align}
	In this region of $\beta$, Corollary~\ref{cor:peierls}
	and Lemma~\ref{lem:damped:naturally} ensure that
	$\phi^\order_{p_\beta}\{
		X: \xpartial X\cap\Gamma_i\neq\varnothing
	\}$ can be made arbitrarily small by tuning $\tau$.
	Hence, for every $\varepsilon>0$, choosing $q+r$ large enough,
	we have $\mu^k_\beta(\left\{\sigma:\sigma_i= k\right\})>1-\varepsilon$.
	
	The measures $\mu^\free_\beta$ and $\phi^\disorder_{p_\beta}$
	can also be coupled, in such a way that,
	given a configuration of bonds, the colour of
	the isolated sites are chosen independently and uniformly
	among the $q+r$ possibilities.
	Using this coupling, and conditioning on whether a given site $i$
	is isolated or not, we obtain
	\begin{align}
		\mu^\free_\beta(\left\{\sigma:\sigma_i=k\right\}) &\leq
			\frac{1}{q+r} +
			\phi^\disorder_{p_\beta}\{
				X: \text{$i$ not isolated in $(\xS,X)$}
			\} \;.
	\end{align}
	If $\beta\leq\beta_\critical$,
	the bond configuration almost surely consists of a unique sea of disorder with finite
	islands of order.  Hence, the probability that site $i$
	is not isolated is bounded by the probability that
	site $i$ is surrounded by a disorder contour; that is,
	\begin{align}
		\phi^\disorder_{p_\beta}\{
			X: \text{$i$ not isolated in $(\xS,X)$}
		\} &\leq
			\phi^\disorder_{p_\beta}\{
				X: \xpartial X\cap\Gamma_i\neq\varnothing
			\} \;.
	\end{align}
	As in the previous case, Corollary~\ref{cor:peierls}
	and Lemma~\ref{lem:damped:naturally} guarantee that
	for every $\varepsilon>0$,
	choosing $q+r$ large enough,
	$\mu^\free_\beta(\left\{\sigma:\sigma_i=k\right\})<\varepsilon$.
	
	It remains to show that for $\beta<\beta_\critical$,
	the $(q,r)$-Potts Gibbs measure is unique.
	
	As in the standard random-cluster model,
	there exists a critical value $0<p_\critical<1$ such that
	\begin{itemize}
		\item for $p<p_\critical$, almost surely with respect to
			$\phi^\order_p$ and $\phi^\disorder_p$,
			there is no infinite connected component of bonds
			(order does not ``percolate''), whereas
		\item for $p>p_\critical$, the event that a given site is in
			an infinite connected component
			happens with positive probability
			under both $\phi^\order_p$ and $\phi^\disorder_p$.
	\end{itemize}
	(See Appendix~\ref{apx:rc:properties}.)
	It follows that $p_\critical=p_{\beta_\critical}$.
	Namely, if $p_\beta<p_\critical$, then order does not percolate
	under $\phi^\order_{p_\beta}$.  Therefore, equation~(\ref{eq:thm:sea-of-order})
	does not hold, implying that $\beta<\beta_\critical$.
	Conversely, if $p_\beta>p_\critical$, then order percolates
	with positive probability under $\phi^\disorder_{p_\beta}$,
	refuting~(\ref{eq:thm:sea-of-disorder}).
	Hence, we must have $\beta>\beta_\critical$.
	
	On the other hand, for every $\beta$ at which
	\begin{align}
		\phi^\order_{p_\beta}(\text{$\exists$ an infinite connected component of bonds})=0 \;,
	\end{align}
	the measure $\mu^\free_\beta$ is the only Gibbs measure for
	the $(q,r)$-Potts model (see Appendix~\ref{apx:rc:properties}).
	The latter condition is guaranteed whenever $p_\beta<p_\critical$,
	which is equivalent to $\beta<\beta_\critical$.
	Thus the uniqueness of Gibbs measure for $\beta<\beta_\critical$ follows.
\end{proof}

\section{Conclusion}
\label{sec:conclusion}
In this paper, we presented a proof that the two-dimensional
Potts model with $q$ visible colours and $r$ invisible colours
undergoes a first-order phase transition
in temperature accompanied by a $q$-fold symmetry breaking,
provided the number of invisible colours is large enough.
On the other hand, for $r=0$ (no invisible colours),
the model reduces to the standard $q$-colour Potts model,
for which it is known that if $q=2,3,4$, the transition in two dimensions is second-order.
Tamura, Tanaka and Kawashima~\cite{TamTanKaw10,TanTam10,TanTamKaw11}
introduced the Potts model with $r$ invisible colours as
a simple two-dimensional example with short-range interactions
in which, tuning the parameter $r$, the same symmetry breaking could accompany
phase transitions of different orders.
The impossibility to infer the order of the phase transition from
the broken symmetry was already noticed in other examples,
such as the two-dimensional $3$-colour Kac-Potts model~\cite{GobMer07}.
For this model, Gobron and Merola proved that
a $3$-fold symmetry breaking might be accompanied with either
a first-order or a second-order phase transition,
by changing the finite range of the interactions.

The first-order phase transition in the $(q,r)$-Potts model
occurs as long as $q+r$ is large enough.
In particular, even for small values of $q$ (say, $q=1,2,3,4$),
the presence of many invisible colours assures
a first-order transition.
The argument is very similar to the one 
for the standard $q$-Potts model, in which $q$ is required to be large~\cite{LaaMesMirRuiShl91,Kot94}.
The transition point is asymptotically (in $q+r$)
given by $\beta_\critical\approx\frac{1}{2}\log(q+r)$.
For $q+r$ large, the latent heat is approximately given
by
$2\left(-\frac{\partial e(\varnothing)}{\partial\beta}+\frac{\partial e(\xB)}{\partial\beta}\right)=
2+\frac{2}{\sqrt{q+r}}$, which tends to $2$ as $q+r\to\infty$.

The proof relies on a formulation of the Potts model with invisible colours
in terms of a variant of the random-cluster model,
which we named the biased random-cluster model.
The difference between this new model and the original
random-cluster model is that it weights singleton connected components
differently from non-singleton connected components.
Such a disparity allows one to increase the entropy
by increasing the number of invisible colours,
while keeping the number of ground states (i.e., the number of visible colours) unchanged.
The random-cluster representation allows for a clear formulation
of order and disorder:
order is associated with the presence of bonds while
disorder with the absence of bonds.
This leads to a simple notion of contours describing
the interface between order and disorder.
Hence, the random-cluster representation
lends itself to a Pirogov-Sinai analysis,
which is used to prove the existence of a first-order phase transition.

We remark that the above analysis extends to higher dimensions.

\appendix

\section{Appendix}

\subsection{Derivation of the Biased Random-Cluster Representation}
\label{apx:random-cluster}


To derive the relation~(\ref{eq:potts-rc:disordered}), we start from~(\ref{eq:potts-rc:finite}) and write
\begin{align}
	Z_\beta(\Lambda_n) &=
		\xe^{\beta\abs{B(\Lambda_n)}}\cdot Z^{\RC}_{p_\beta,q,r}(\Lambda_n) \\
	&=
		\sum_{X\subseteq B(\Lambda_n)}
			(\xe^\beta-1)^{\abs{X}}
			(q+r)^{\kappa_0(S(\Lambda_n),X)}
			q^{\kappa_1(S(\Lambda_n),X)} \\
	&=
		(q+r)^{-\abs{S(\Lambda_{n+1}\setminus\Lambda_n)}}
		\hspace{-2em}
		\sum_{\substack{Y\subseteq B(\Lambda_{n+1}) \\
			Y\cap\left(B(\Lambda_{n+1})\setminus B(\Lambda_n)\right) =\varnothing
		}}
		\hspace{-2em}
			(\xe^\beta-1)^{\abs{Y}}
			(q+r)^{\kappa_0(S(\Lambda_{n+1}),Y)}
			q^{\kappa_1(S(\Lambda_{n+1}),Y)} \\
	&=
		(q+r)^{-\abs{S(\Lambda_{n+1}\setminus\Lambda_n)}}
		\cdot\xe^{\beta\abs{B(\Lambda_{n+1})}}\cdot
		Z^{\RC.\disorder}_{p_\beta,q,r}(\Lambda_{n+1}) \;.
\end{align}


To obtain the relation~(\ref{eq:potts-rc:ordered}), we need to take
the homogeneous boundary condition for the $(q,r)$-Potts model
into account.
Denoting the set of $(q,r)$-Potts configurations on $\Lambda_n$
by $\Omega_{\Lambda_n}$, 
we start from the definition~(\ref{eq:potts:partition:boundary})
and write
\begin{align}
	Z_\beta^{\omega^k}(\Lambda_n) &=
		 \sum_{\sigma\in\Omega_{\Lambda_n}}\exp\left\lbrace 
		 	\;
			\beta\hspace{-1em}\sum_{\{i,j\}\in B(\Lambda_n)}\hspace{-1em} \delta(\sigma_i=\sigma_j\leq q)
		    \hspace{2em}+\hspace{2em}\beta\hspace{-2em}\sum_{\substack{
			\{i,j\}\in B(\Lambda_{n+1})\\
			i\in S(\Lambda_n),\;
			j \notin S(\Lambda_n)}
			}\hspace{-2em} \delta(\sigma_i=\omega^k_j= k)
			\;
		\right\rbrace 
	 \\
	&=
	   \sum_{\sigma\in\Omega_{\Lambda_n}}\;\prod_{\{i,j\}\in B(\Lambda_n)} \hspace{-0.7em}\left( 
	   1+\delta(\sigma_i=\sigma_j\leq q)(\xe^\beta -1)
	\right) 
	   \prod_{\substack{\{i,j\}\in B(\Lambda_{n+1})\\
	   i\in S(\Lambda_n),\;
	   j\notin S(\Lambda_n)}
	   }\hspace{-1em}\left( 
	   1+\delta(\sigma_i=k)(\xe^\beta -1)
	 \right) \;.
\end{align}
Denoting
\begin{align}
	\partial\Lambda_n &\IsDef \left\{
			\{i,j\}\in B(\Lambda_{n+1}): i\in S(\Lambda_n) \text{ and } j\notin S(\Lambda_n)
		\right\} \;,
\end{align}
we can expand the products to obtain
\begin{align}
	Z_\beta^{\omega^k}(\Lambda_n)
	&=\sum_{\sigma\in\Omega_{\Lambda_n}}
		\sum_{X_1\subseteq B(\Lambda_n)}\sum_{X_2\subseteq\partial\Lambda_n}
		(\xe^\beta -1)^{\abs{X_1}+\abs{X_2}} \; \delta(\sigma\in \Xi_k(X_1,X_2)) \;,
\end{align}
where
\begin{align}
	&\Xi_k(X_1,X_2) \nonumber\\
	&\IsDef \left\{
		\sigma\in\Omega_{\Lambda_n}: 
		\sigma_i=\sigma_j\leq q \text{ for all $\{i,j\}\in X_1$ and }
		\sigma_i=k 
			\text{ for all $i\in S(\Lambda_n)\cap S(X_2)$}
	\right\} \;.
\end{align}
To impose the ordered boundary condition,
we multiply and divide by
$(\xe^\beta-1)^{\abs{B(\Lambda_{n+1}\setminus\Lambda_n)}}$
to emulate the presence of the bonds in $B(\Lambda_{n+1}\setminus\Lambda_n)$.
This gives
\begin{align}
	Z_\beta^{\omega^k}(\Lambda_n) &=
		(\xe^\beta-1)^{-\abs{B(\Lambda_{n+1}\setminus\Lambda_n)}}
		\sum_{\sigma\in\Omega_{\Lambda_n}}
		\sum_{
			\substack{
				X\subseteq B(\Lambda_{n+1})\\
				X\supseteq B(\Lambda_{n+1}\setminus\Lambda_n)
			}
		}
		(\xe^\beta -1)^{\abs{X}} \; \delta(\tilde{\sigma}\in \Theta_k(X)) \;,
\end{align}
where $\tilde{\sigma}$ is the extension of $\sigma$ to
a configuration in $\Omega_{\Lambda_{n+1}}$ with
$\tilde{\sigma}_i=k$ for $i\in S(\Lambda_{n+1}\setminus\Lambda_n)$
and
\begin{align}
	&\Theta_k(X) \nonumber\\
	&\IsDef
	\left\{
		\tilde{\sigma}\in\Omega_{\Lambda_{n+1}}: 
		\tilde{\sigma}_i=\tilde{\sigma}_j\leq q \text{ for all $\{i,j\}\in X$ and }
		\tilde{\sigma}_i=k \text{ for all $i\in S(\Lambda_{n+1}\setminus\Lambda_n)$}
	\right\} \;.
\end{align}
Changing the order of the sums gives
\begin{align}
	Z_\beta^{\omega^k}(\Lambda_n) &=
		(\xe^\beta-1)^{-\abs{B(\Lambda_{n+1}\setminus\Lambda_n)}}
		\hspace{-1em}
		\sum_{
			\substack{
				X\subseteq B(\Lambda_{n+1})\\
				X\supseteq B(\Lambda_{n+1}\setminus\Lambda_n)
			}
		}
		\hspace{-1em}
		(\xe^\beta -1)^{\abs{X}} \;
		\abs{\Theta_k(X)} \;.
\end{align}
Note that for $X$
satisfying $B(\Lambda_{n+1}\setminus\Lambda_n)\subseteq X\subseteq B(\Lambda_n)$,
the size of $\Theta_k(X)$ is
\begin{align}
	q^{-1}\cdot(q+r)^{\kappa_0(S(\Lambda_{n+1}),X)}q^{\kappa_1(S(\Lambda_{n+1}),X)}
\end{align}
we obtain
\begin{align}
	Z_\beta^{\omega^k}(\Lambda_n) &=
		q^{-1}(\xe^\beta-1)^{-\abs{B(\Lambda_{n+1}\setminus\Lambda_n)}}
		\hspace{-1em}
		\sum_{
			\substack{
				X\subseteq B(\Lambda_{n+1})\\
				X\supseteq B(\Lambda_{n+1}\setminus\Lambda_n)
			}
		}
		\hspace{-1em}
		(\xe^\beta -1)^{\abs{X}}\cdot
		(q+r)^{\kappa_0(S(\Lambda_{n+1}),X)}q^{\kappa_1(S(\Lambda_{n+1}),X)} \\
	&=
		q^{-1}\cdot(\xe^\beta-1)^{-\abs{B(\Lambda_{n+1}\setminus\Lambda_n)}}
		\cdot\xe^{\beta\abs{B(\Lambda_{n+1})}}\cdot
		Z^{\RC.\order}_{p_\beta,q,r}(\Lambda_{n+1}) \;.
\end{align}


\subsection{Derivation of the Contour Representation}
\label{apx:contour}
We show that for $\Lambda=\Lambda_{n+1}$,
the definitions~(\ref{eq:rc-partition-function:disordered}) and~(\ref{eq:rc-disordered-def:old})
(resp., (\ref{eq:rc-partition-function:ordered}) and~(\ref{eq:rc-ordered-def:old})) agree.

The weight of a configuration $X\subseteq B(\Lambda_{n+1})$ is
\begin{align}
	p_\beta^{\abs{X}}(1-p_\beta)^{\abs{B(\Lambda_{n+1})\setminus X}}
		(q+r)^{\kappa_0(S(\Lambda_{n+1}),X)}q^{\kappa_1(S(\Lambda_{n+1}),X)} \;.
\end{align}
For a configuration $X$ in $\mathcal{X}^\disorder_{\Lambda_{n+1}}$ or
$\mathcal{X}^\order_{\Lambda_{n+1}}$, let $\partial_X$ be the corresponding
contour family in $\Delta_{\Lambda_{n+1}}^\disorder$ or $\Delta_{\Lambda_{n+1}}^\order$.
More precisely, $\partial_X\IsDef\xpartial X$ if $\mathcal{X}^\disorder_{\Lambda_{n+1}}$ and
$\partial_X\IsDef \partial\left(X\cup B(\Lambda_{n+1})^\complement\right)$ if $\mathcal{X}^\order_{\Lambda_{n+1}}$.
\begin{claim}
	For $X\in\mathcal{X}^\disorder_{\Lambda_{n+1}}$ we have the relation
	\begin{align}
	\label{eq:relation:bond-site:d}
		2\abs{B(\Lambda_{n+1})\setminus X} &=
			4\kappa_0(S(\Lambda_{n+1}),X) + \sum_{\gamma\in\partial_X}\abs{\gamma} - \abs{\xpartial B(\Lambda_{n+1})} \;,
	\end{align}
	and for $X\in\mathcal{X}^\order_{\Lambda_{n+1}}$ we have
	\begin{align}
	\label{eq:relation:bond-site:o}
		2\abs{B(\Lambda_{n+1})\setminus X} &=
			4\kappa_0(S(\Lambda_{n+1}),X) + \sum_{\gamma\in\partial_X}\abs{\gamma} \;.
	\end{align}
\end{claim}
\begin{proof}
	We first decompose the set
	$\left\{(i,b): b\in B(\Lambda_{n+1})\setminus X \text{ and } i\sim b \right\}$ as
	\begin{gather}
		\left\{(i,b): b\in B(\Lambda_{n+1})\setminus X \text{ and } i\sim b \text{ and } i\in S(X) \right\} \\
		\cup \nonumber \\
		\label{eq:appx:set:2}
		\;\;\left\{(i,b): b\in B(\Lambda_{n+1})\setminus X \text{ and } i\sim b \text{ and } i\notin S(X) \right\} \;,
	\end{gather}
	and furthermore note that the latter set can be expressed as
	\begin{gather}
		\left\{(i,b): i\in S(\Lambda_{n+1})\setminus S(X) \text{ and } i\sim b\right\} \\
		\setminus \nonumber \\
		\left\{(i,b): i\in S(\Lambda_{n+1})\setminus S(X) \text{ and } i\sim b \text{ and }
			b\notin B(\Lambda_{n+1})\right\} \;.
	\end{gather}
	We have
	\begin{align}
		\abs{\left\{(i,b): b\in B(\Lambda_{n+1})\setminus X \text{ and } i\sim b \right\}} &=
			2\abs{B(\Lambda_{n+1})\setminus X} \;, \\
		\abs{\left\{(i,b): b\in B(\Lambda_{n+1})\setminus X \text{ and } i\sim b \text{ and } i\in S(X) \right\}} &=
			\sum_{\gamma\in\partial_X}\abs{\gamma} \;, \\
		\abs{\left\{(i,b): i\in S(\Lambda_{n+1})\setminus S(X) \text{ and } i\sim b\right\}} &=
			4\kappa_0(S(\Lambda_{n+1}),X) \;,
	\end{align}
	and the cardinality of
	$\left\{(i,b): i\in S(\Lambda_{n+1})\setminus S(X) \text{ and } i\sim b \text{ and }
			b\notin B(\Lambda_{n+1})\right\}$ equals
	\begin{align}
			\begin{cases}
				\abs{\xpartial B(\Lambda_{n+1})}\;, \quad & \text{if $X\in\mathcal{X}_{\Lambda_{n+1}}^\disorder$,} \\
				0\;, & \text{if $X\in\mathcal{X}_{\Lambda_{n+1}}^\order$.}
			\end{cases}
	\end{align}
\end{proof}

Using the relations~(\ref{eq:relation:bond-site:d}) and~(\ref{eq:relation:bond-site:o})
the weight of $X$ takes the form
\begin{align}
	(q+r)^{\frac{1}{4}\abs{\xpartial B(\Lambda_{n+1})}}
	\xe^{-e(\xB)\cdot\abs{X}}\xe^{-e(\varnothing)\cdot\abs{B(\Lambda_{n+1})\setminus X}}
		(q+r)^{-\frac{1}{4}\sum_{\gamma\in\partial_X}\abs{\gamma}} q^{\kappa_1(S(\Lambda_{n+1}),X)}
\end{align}
if $X\in\mathcal{X}_{\Lambda_{n+1}}^\disorder$ and
\begin{align}
	\xe^{-e(\xB)\cdot\abs{X}}\xe^{-e(\varnothing)\cdot\abs{B(\Lambda_{n+1})\setminus X}}
		(q+r)^{-\frac{1}{4}\sum_{\gamma\in\partial_X}\abs{\gamma}} q^{\kappa_1(S(\Lambda_{n+1}),X)}
\end{align}
if $X\in\mathcal{X}_{\Lambda_{n+1}}^\order$.
If $X\in\mathcal{X}_{\Lambda_{n+1}}^\disorder$, every non-singleton connected component
in $(S(\Lambda_{n+1}),X)$ contains all the sites of a unique disorder contour in $\partial_X$, so that
the number of disorder contours in $\partial_X$ is the same as $\kappa_1(S(\Lambda_{n+1}),X)$,
and we have
\begin{align}
	(q+r)^{-\frac{1}{4}\sum_{\gamma\in\partial_X}\abs{\gamma}} q^{\kappa_1(S(\Lambda_{n+1}),X)} &=
		\prod_{\gamma\in\partial_X}\rho(\gamma) \;.
\end{align}
On the other hand, if $X\in\mathcal{X}_{\Lambda_{n+1}}^\order$, the outermost
non-singleton connected component of $(S(\Lambda_{n+1}),X)$ has no
associated disorder contour in $\partial_X$, thus
\begin{align}
	(q+r)^{-\frac{1}{4}\sum_{\gamma\in\partial_X}\abs{\gamma}} q^{\kappa_1(S(\Lambda_{n+1}),X)} &=
		q\prod_{\gamma\in\partial_X}\rho(\gamma) \;.
\end{align}

In conclusion, summing over all configurations, we obtain
\begin{align}
	Z^{\RC.\disorder}(\Lambda_{n+1}) &=
		(q+r)^{\frac{\abs{\xpartial B(\Lambda_{n+1})}}{4}}
		\sum_{X\in\mathcal{X}_{\Lambda_{n+1}}^\disorder}
		\xe^{
			-\abs{X}\cdot e(\xB)
			-\abs{B(\Lambda_{n+1})\setminus X}\cdot e(\varnothing)
		}
		\prod_{\gamma\in\partial}\rho(\gamma) \;, \\
	Z^{\RC.\order}(\Lambda_{n+1}) &=
		q\sum_{X\in\mathcal{X}_{\Lambda_{n+1}}^\order}
		\xe^{
			-\abs{X}\cdot e(\xB)
			-\abs{B(\Lambda_{n+1})\setminus X}\cdot e(\varnothing)
		}
		\prod_{\gamma\in\partial}\rho(\gamma) \;.
\end{align}

We remark that the definitions~(\ref{eq:rc-partition-function:disordered})
and~(\ref{eq:rc-partition-function:ordered}) may also be extended
to general finite volumes of the lattice in a compatible fashion.
Namely, if for a volume $\Lambda$, we define 
\begin{align}
	\mathcal{X}_\Lambda^\disorder &\IsDef \left\{
		X\subseteq B(\Lambda): \xpartial X\in\Delta_\Lambda^\disorder
	\right\} \;, \\
	\mathcal{X}_\Lambda^\order &\IsDef \left\{
		X\subseteq B(\Lambda): \xpartial\left(X\cup B(\Lambda)^\complement\right)\in\Delta_\Lambda^\order
	\right\} \;,
\end{align}
the compatibility of the definitions can be verified similarly.

\subsection{A Few Properties of the Biased Random-Cluster Model}
\label{apx:rc:properties}
Much information about the standard Potts model can be detected
by studying the corresponding random-cluster model (see~\cite{GeoHagMae00,Gri06,Gri10}).
Many of the properties of the standard random-cluster model
can be extended to the biased random-cluster model.
These properties, in turn, can be used, in a similar fashion,
to obtain information about the Potts model with invisible colours.
In this appendix, we briefly sketch some of these properties
that we exploit in the proof of the main theorem.
The proofs are straightforward modifications of the
standard case, which can be found in~\cite{GeoHagMae00,Gri06,Gri10}.

Let $\xG=(S,B)$ be a finite graph.
The configurations of the biased random-cluster model
on~$\xG$ can be ordered according to the inclusion ordering.
A configuration $X\subseteq B$ is considered to be
\emph{smaller than or equal to}
a configuration $Y\subseteq B$, if and only if
every bond present in~$X$ is also present in~$Y$.
An event $\mathcal{E}\subseteq 2^B$ is \emph{increasing}
if for every two configurations $X$ and $Y$
such that $X\in\mathcal{E}$ and $Y\supseteq X$,
we have $Y\in\mathcal{E}$.
We say that a probability distribution
$\nu$ is \emph{positively correlated},
if $\nu(\mathcal{E}_1\cap \mathcal{E}_2) \geq \nu(\mathcal{E}_1)\nu(\mathcal{E}_2)$.
The inclusion ordering on the configuration space~$2^B$ induces an ordering
on the space of probability distributions on $2^B$.
If $\nu_1$ and $\nu_2$ are probability distributions on $2^B$,
we write $\nu_1\preceq\nu_2$ if $\nu_1(\mathcal{E})\leq\nu_2(\mathcal{E})$
for every increasing event $\mathcal{E}\subseteq 2^B$.
In this case, we say that $\nu_1$ is
\emph{stochastically dominated} by $\nu_2$.

For every $0<p<1$, $q\geq 1$ and $r\geq 0$,
the $r$-biased random-cluster distribution $\phi_{p,q,r}$
on~$\xG$ is positively correlated.
This follows from the Fortuin-Kasteleyn-Ginibre theorem
(Theorem~4.11 of~\cite{GeoHagMae00}; see Corollary~6.7).
It follows that,
if $\mathcal{E}$ is an increasing (resp., decreasing)
event with $\phi_{p,q,r}(\mathcal{E})>0$,
then
the conditional distribution
$\phi_{p,q,r}(\cdot\,|\,\mathcal{E})$
stochastically dominates
(resp., is dominated by) $\phi_{p,q,r}$.
Furthermore, if $0<p_1\leq p_2<1$,
it follows from Holley's theorem
(Theorem~4.8 of~\cite{GeoHagMae00})
that $\phi_{p_1,q,r}\preceq\phi_{p_2,q,r}$
(see Corollary~6.7 in~\cite{GeoHagMae00}).

Let $\Lambda$ be a finite volume in the lattice
and $\phi_\Lambda$ the biased random-cluster
distribution on $\Lambda$ (as a graph, without boundary condition).
Let us denote by $\phi^\order_\Lambda$ and $\phi^\disorder_\Lambda$
the biased random-cluster distributions on $\Lambda$
with ordered and disordered boundary conditions, respectively.
By an application of the positive correlation property
of $\phi_\Lambda$ we have
\begin{align}
	\phi^\disorder_\Lambda \preceq \phi_\Lambda \preceq \phi^\order_\Lambda \;.
\end{align}
Moreover, by a further application of
the Fortuin-Kasteleyn-Ginibre theorem,
the distributions
$\phi^\disorder_\Lambda$ and $\phi^\order_\Lambda$
are also positively correlated.
This implies that
if~$\Lambda_1$ is a sub-volume of~$\Lambda_2$, we have
\begin{align}
	\phi^\disorder_{\Lambda_1} \preceq \phi^\disorder_{\Lambda_2}
	\qquad\text{ and }\qquad
	\phi^\order_{\Lambda_1} \succeq \phi^\order_{\Lambda_2} \;.
\end{align}
As in Lemma~6.8 of~\cite{GeoHagMae00},
this implies that the weak limits
\begin{align}
	\phi^\disorder \IsDef\lim_{\Lambda\uparrow\xL}\phi^\disorder_\Lambda
	\qquad\text{ and }\qquad
	\phi^\order \IsDef\lim_{\Lambda\uparrow\xL}\phi^\order_\Lambda
\end{align}
exist, where the limit $\Lambda\uparrow\xL$ can be taken along
the net of all finite volumes in $\xL$ with the inclusion ordering.

To emphasize the dependence on parameter $p$,
let us write $\phi^\order_{\Lambda,p}$ and $\phi^\disorder_{\Lambda,p}$
for the biased random-cluster distributions with parameter $p$.
Then, by an application of Holley's theorem,
if $0<p_1\leq p_2<1$, we have
\begin{align}
	\phi^\disorder_{\Lambda,p_1} \preceq \phi^\disorder_{\Lambda,p_2}
	\qquad\text{ and }\qquad
	\phi^\order_{\Lambda,p_1} \preceq \phi^\order_{\Lambda,p_2} \;.
\end{align}
Let $i\operc\infty$ denote the event that 
there exists an infinite path of bonds passing through site~$i$
(``order'' \emph{percolates} from site $i$ to infinity).
The latter stochastic inequalities imply that the probabilities
$\phi^\order_p(i\operc\infty)$ and $\phi^\disorder_p(i\operc\infty)$
are increasing in $p$.
This monotonicity assures the existence of \emph{critical probabilities}
$0\leq p^\order_\critical,p^\disorder_\critical\leq 1$ such that
for every $p<p^\order_\critical$
we have $\phi^\order_p(i\operc\infty) = 0$
while for every $p>p^\order_\critical$
we have $\phi^\order_p(i\operc\infty) > 0$,
and similarly for $p^\disorder_\critical$.
The critical probabilities are given by
\begin{align}
	p^\order_\critical &\IsDef
		\sup\left\{p: \phi^\order_p(i\operc\infty) = 0 \right\} \;, \\
	p^\disorder_\critical &\IsDef
		\sup\left\{p: \phi^\disorder_p(i\operc\infty) = 0 \right\} \;.
\end{align}

It turns out that the two critical probabilities are actually the same,
hence we define $p_\critical\IsDef p^\order_\critical = p^\disorder_\critical$.
This follows from the fact that the probability measures
$\phi^\order_p$ and $\phi^\disorder_p$ may differ for
at most countably many values of $p$.
The latter can be proved in a very similar manner
as done in Theorem~8.17 of~\cite{Gri10}
for the standard random-cluster measures.

By means of the coupling, many properties of the $(q,r)$-Potts measures
can be derived from the corresponding $r$-biased random-cluster measures.
For instance, one can show that the thermodynamic limits $\mu^k$ and
$\mu^\free$ do not depend on the sequence $\left\{\Lambda_n\right\}_n$ of volumes
along which the limits are taken.
In fact, the limits
\begin{align}
	\mu^k=\lim_{\Lambda\uparrow\xL}\mu^k_\Lambda
	\qquad\text{ and }\qquad
	\mu^\free=\lim_{\Lambda\uparrow\xL}\mu^\free_\Lambda
\end{align}
can be taken along the net of all finite volumes in $\xL$.
In particular, this implies the translation-invariance of $\mu^k$ and $\mu^\free$.
The proofs are similar to those of the standard case (Proposition~6.9 of~\cite{GeoHagMae00}).

Uniqueness and multiplicity of the $(q,r)$-Potts measures are related to
the percolation of ``order'' in the $r$-biased random-cluster model.
More specifically,
if $\phi^\order(\text{order percolate})=0$, then the $(q,r)$-Potts model
admits a unique Gibbs measure (as in Theorem~6.10 in~\cite{GeoHagMae00}).
On the other hand, if
\begin{align}
	\phi^\order\left(
		\begin{array}{c}
			\text{$\exists$ unique infinite sea of order} \\
			\text{with finite islands of disorder}
		\end{array}
	\right)=1 \;,
\end{align}
then the measures $\mu^1,\mu^2,\ldots,\mu^q$ are distinct
and satisfy
\begin{align}
	\mu^k\left(
		\begin{array}{c}
			\text{$\exists$ unique infinite uni-colour sea,} \\
			\text{which has colour $k$}
		\end{array}
	\right)=1 \;.
\end{align}
(Recall that a ``sea'' of order in a random-cluster configuration is simply
a connected component of bonds.  A ``uni-colour sea'' in a Potts configuration
refers to a maximal connected subgraph of the lattice
induced by sites having the same colour.) 

The latter claim is a consequence of the existence of a coupling
between $\phi^\order$ and $\mu^k$ (for visible~$k$),
which can be constructed as follows:
\begin{enumerate}[ i)]
	\item We first sample a bond configuration $X$ according to $\phi^\order$.
	\item For every site $i$ that $i\operc\infty$ in $(\xS,X)$, we colour $i$ with colour $k$.
	\item For every finite non-singleton connected component of $(\xS,X)$, we choose a random visible colour uniformly
		among the $q$ possibilities and colour all the sites in the component with this colour.
	\item For every isolated site $i$ in $(\xS,X)$, we choose a random colour uniformly among the $q+r$ possible colours.
\end{enumerate}
The fact that the marginal of this construction on spin configurations is $\mu^k$
is parallel to Theorem~4.91 in~\cite{Gri06} and has a similar proof.
An analogous coupling exists between $\phi^\disorder$ and $\mu^\free$.

\bibliographystyle{plain}
\bibliography{files/bibliography}

\end{document}

%% file: macros.tex
\usepackage{graphicx}
\usepackage{color}
\usepackage{amssymb}
\usepackage{amsthm}
\usepackage{amsmath}
\usepackage{mathrsfs}
\usepackage{mathtools}	
\usepackage{wasysym}
\usepackage{braket}		
\usepackage{enumerate}
\usepackage{calc}
\usepackage{ifthen}
\usepackage[normalem]{ulem}
\usepackage{geometry}
\usepackage{tikz}
\usepgflibrary{arrows}
\usepackage{graphicx}

\theoremstyle{plain}
\newtheorem*{theorem}{Theorem}	
\newtheorem{lemma}{Lemma} 
\newtheorem{corollary}{Corollary}[lemma]		
\newtheorem{proposition}{Proposition} 
\newtheorem*{claim}{Claim}

\newcommand{\xZ}{\mathbb{Z}}			
\newcommand{\xR}{\mathbb{R}}			
\newcommand{\xL}{\mathbb{L}}			
\newcommand{\xG}{\mathbb{G}}			
\newcommand{\xB}{\mathbb{B}}			
\newcommand{\xS}{\mathbb{S}}			
\newcommand{\IsDef}{\triangleq}			
\newcommand{\abs}[1]{
	\left\lvert#1\right\rvert%
}

\newcommand{\xd}{\mathrm{d}}			
\newcommand{\xe}{\mathrm{e}}			
\newcommand{\interior}{
	\operatorname{\mathrm{int}}%
}
\newcommand{\exterior}{
	\operatorname{\mathrm{ext}}%
}
\newcommand{\diam}{
	\operatorname{\mathrm{diam}}%
}
\newcommand{\xpartial}{\overline{\partial}}

\newcommand{\free}{{\rm free}}

\newcommand{\order}{{\rm ord}}
\newcommand{\disorder}{{\rm disord}}
\newcommand{\xorder}{{\rm o}}
\newcommand{\xdisorder}{{\rm d}}
\newcommand{\RC}{{\sf RC}}

\newcommand{\extern}{{\rm ext}}

\newcommand{\SMALL}{{\rm small}}
\newcommand{\critical}{{\rm c}}
\newcommand{\operc}{\xleftrightarrow{\,\xorder\,}}		

\allowdisplaybreaks[3] 			

\makeatletter
\renewcommand\maketag@@@[1]{%
	\,\rlap{\hspace{\marginparsep}\m@th\normalfont\footnotesize#1}%
}
\makeatother

\makeatletter
\renewcommand{\@seccntformat}[1]{%
	\llap{\csname the#1\endcsname\hspace{\marginparsep}}%
}
\makeatother

\makeatletter
\renewcommand{\@biblabel}[1]{%
	\llap{[#1]\hspace{\marginparsep}}\!\!\!%
}
\makeatother